\newif\ifwww
\newif\ifarxiv
\newif\ifcolor
\newif\ifhide
\wwwfalse
\arxivtrue
\colorfalse
\hidetrue

\ifwww
\documentclass[sigconf]{acmart}
\copyrightyear{2019}
\acmYear{2019} 
\setcopyright{iw3c2w3}
\acmConference[WWW '19]{Proceedings of the 2019 World Wide Web Conference}{May 13--17, 2019}{San Francisco, CA, USA}
\acmBooktitle{Proceedings of the 2019 World Wide Web Conference (WWW '19), May 13--17, 2019, San Francisco, CA, USA}
\acmPrice{}
\acmDOI{10.1145/3308558.3313692}
\acmISBN{978-1-4503-6674-8/19/05}
\else
\documentclass{article}

\usepackage{arxiv}

\usepackage[utf8]{inputenc} 
\usepackage[T1]{fontenc}    
\usepackage{hyperref}       
\usepackage{url}            
\usepackage{booktabs}       
\usepackage{amsfonts}       
\usepackage{nicefrac}       
\usepackage{microtype}      
\usepackage{lipsum}
\fi

\input{setup.tex}

\begin{document}

\newtheorem{claim}{Claim}
\newtheorem{observation}[claim]{Observation}

\ifarxiv
\newtheorem{theorem}{Theorem}
\newtheorem{definition}{Definition}
\newtheorem{example}{Example}
\newtheorem{lemma}{Lemma}
\newtheorem{corollary}[claim]{Corollary}
\newtheorem{proposition}[claim]{Proposition}
\fi

\arxiv{\title{Selling a Single Item with Negative Externalities\\To Regulate Production or Payments?}}

\ifarxiv
\author{
  Tithi Chattopadhyay\\
  Center for Information Technology Policy\\
  Princeton University\\
  Princeton, NJ 08540 \\
  \texttt{tithic@princeton.edu}\\
    \And
    Nick Feamster\\
    Department of Computer Science\\
    Princeton University\\
    Princeton, NJ 08540 \\
    \texttt{feamster@princeton.edu}
    \And
    Matheus V. X. Ferreira\\
    Department of Computer Science\\
    Princeton University\\
    Princeton, NJ 08540 \\
    \texttt{mvxf@cs.princeton.edu}
    \And
    Danny Yuxing Huang\\
    Department of Computer Science\\
    Princeton University\\
    Princeton, NJ 08540 \\
    \texttt{yuxingh@cs.princeton.edu}
    \And
    S. Matthew Weinberg\\
    Department of Computer Science\\
    Princeton University\\
    Princeton, NJ 08540 \\
    \texttt{smweinberg@princeton.edu}
}
\fi

\newcommand{\indicator}[1]{\mathbb I(#1)}
\newcommand{\expectation}[2]{\mathbb E_{#1}[#2]}
\newcommand{\rev}{\textsc{Prof}}
\newcommand{\ext}{\textsc{Ext}}
\newcommand{\risk}{\textsc{risk}}
\newcommand{\hygiene}{h^*}
\newcommand{\pvalue}{\textsc{value}}
\newcommand{\effort}{\textsc{effort}}
\newcommand{\supp}{\textsc{supp}}
\newcommand{\APPROX}{\textsc{Approx}}
\newcommand{\SKILL}{\textsc{Skill}}
\newcommand{\FINE}{\textsc{Fine}}
\newcommand{\BLOWUP}{\textsc{Blowup}}
\newcommand{\COST}{\textsc{Cost}}
\newcommand{\INV}{\textsc{Inv}}
\newcommand{\HEAVY}{\textsc{Heavy}}
\newcommand{\loss}{\ensuremath{\ell}}

\newcommand{\danny}[1]{\textcolor{blue}{#1}}

\newcommand{\nick}[1]{ \textcolor{red}{NF:} \textcolor{blue}{[#1]}}
\newcommand{\mattnote}[1]{\textcolor{blue}{#1}}

\newcommand{\matheusnote}[1]{
\ifhide
\else
\textcolor{orange}{#1}
\fi
}

 \newcommand{\hide}[1]{{\iffalse #1 \fi}}
 
\newcommand{\short}[1]{#1}
\newcommand{\complete}[1]{{\iffalse #1 \fi}}

\newcommand{\note}[1]{\textcolor{red}{#1}}

\arxiv{\maketitle}

\begin{abstract}
We consider the problem of regulating products with negative externalities to a third party that is neither the buyer nor the seller, but where both the buyer and seller can take steps to mitigate the externality. The motivating example to have in mind is the sale of Internet-of-Things (IoT) devices, many of which have historically been compromised for DDoS attacks that disrupted Internet-wide services such as Twitter~\cite{krebs, guardian}. Neither the buyer (i.e., consumers) nor seller (i.e., IoT manufacturers) was known to suffer from the attack, but both have the power to expend effort to secure their devices. We consider a regulator who regulates payments (via fines if the device is compromised, or market prices directly), or the product directly via mandatory security requirements.

Both regulations come at a cost---implementing security requirements increases production costs, and the existence of fines decreases consumers' values---thereby reducing the seller's profits. The focus of this paper is to understand the \emph{efficiency} of various regulatory policies. That is, policy A is more efficient than policy B if A more successfully minimizes negatives externalities, while both A and B reduce seller's profits equally.

We develop a simple model to capture the impact of regulatory policies on a buyer's behavior.  {In this model, we show that for \textit{homogeneous} markets---where the buyer's ability to follow security practices is always high or always low---the optimal (externality-minimizing for a given profit constraint) regulatory policy need regulate \emph{only} payments \emph{or} production.}
In arbitrary markets, by contrast, we show that while the optimal policy may require regulating both aspects, there is always an approximately optimal policy which regulates just one.
\end{abstract}


\www{\maketitle}

\matheusnote{Purple is only rendered on the arxiv version. Brown is only rendered on the www version. (the flags are set on top of paper.tex)}
\sloppy
\section{Introduction}\label{sec-introduction}
The Tragedy of the Commons is a well-documented phenomenon where agents act in their own personal interests, but their collective action brings detriments to the common good~\cite{hardin1968tragedy}. One motivating example that we will keep referencing in the paper is the sale of Internet-of-Things (IoT) devices, such as Internet-connected cameras, light bulbs, and refrigerators. Recent years have seen a proliferation of these ``smart-home'' devices, many of which are known to contain security vulnerabilities that have been exploited to launch high-profile attacks and disrupt Internet-wide services such as Twitter and Reddit~\cite{krebs, guardian}. Both the owners and manufacturers of IoT devices have the ability to protect the common good (i.e., Internet-wide service for all users) from being attacked by securing their devices, but have little incentive to do so. For the manufacturers, implementing security features, such as using encryption or having no default passwords, introduces extra engineering cost~\cite{august2016market}. Similarly, security practices, such as regularly updating the firmware or using complex and difficult-to-remember passwords, can be a costly endeavor for the consumers~\cite{choi2010network,dancingpigs}. The results of their actions cause a negative \textit{externality}, where Internet service is disrupted for other users.

One way to reduce the negative externality is regulation. In the context of IoT sales, a regulator can, for instance, set minimum security standards for the manufacturers or impose fines on owners of hacked IoT devices that engage in attacks. Fines could come in a few forms: direct levies on the consumer, or indirect monetary incentives. For instance, ISPs could offer discounts to users whose networks have not displayed any signs of malicious activities. One might argue that such penalty-based policies could be too futuristic, but it is worth noting that similar practices are being adopted in other industries to mitigate negative externalities~\cite{10.2307/1803378}. One example is the levying of fines on users (such as cars and factories) that cause pollution~\cite{10.2307/2950868}. While there are a lot of practicalities that have to be kept in mind and the decision of when to implement consumer/user fines depends on various factors, this is certainly one of the various policy alternatives that is worthwhile to study. Such regulations, however, can potentially increase the cost of production,  discourage consumers from purchasing IoT devices, and reduce the manufacturer's profit. Our focus is to compare the \textit{efficiency} of various regulatory policies: for two policies which equally hurt the seller's profits, which one better mitigates externalities? We will also be interested in understanding the \emph{optimal} policy: the minimum security standards for the manufactures and fines on owners that together best mitigate externalities subject to a minimum seller's profit.

We first develop a model that consists of a buyer (e.g., consumers interested in purchasing IoT devices) and a single product for sale (e.g., IoT device). The product may come with some (costly to increase) level of security, $c$, and a consumer purchasing the device may choose to spend additional effort $h$ to further secure the device. We consider a mechanism for regulating the market through incentives, for example, by requiring that the product being sold implement security features that cost the seller $c$ dollars, imposing a fine of $y$ dollars on the buyer if the product is later compromised and used in attacks, or both. Which intervention is more appropriate depends on how efficient buyers are in securing the product. The goal of the regulator is to minimize the negative externalities subject to a cap on the negative impact on the seller's profits---the idea being that any policy which too negatively impacts the seller could be unimplementable due to industry backlash.

Understanding the effects of such regulations on the behavior of a \emph{single} consumer is relatively straightforward. For example: as fines go up, consumers adjust (upwards) the optimal level of effort to expend, lowering their total value for the item. Yet, reasoning about how an entire market of consumers will respond to changes, and how these responses impact seller profits becomes more complex. 

Our contributions are as follows: (i)~We model the sale of a single item with negative externalities, using the sale of IoT devices as the motivating example (Section~\ref{sec:model}). (ii)~{We show in Sections~\ref{sec:deterministic-distribution} and~\ref{sec:homogeneous-distribution} that  when the population of consumers is \emph{homogeneous} (i.e., all consumers are comparably effective at translating effort into security) that optimal policies need only to regulate \emph{either} the product (via minimum security standards) \emph{or} the payments (via fines)}. (iii)~We provide an example of non-homogeneous markets where the optimal policy regulates both product and payments, but prove that in \emph{all} markets, it is always approximately optimal to regulate only one (Sections~\ref{sec:lower-bound} and~\ref{sec:approximation}). The technical sections additionally contain numerous examples witnessing the subtleties in reasoning about these problems, and that any assumptions made in our theorem statements are necessary. 

\section{Related Work}\label{sec-related-work}
\hide{\paragraph{Auctions with Externalities}

\cite{jehiel1999multidimensional} introduced a model where externalities consist in a fixed payoff (possibly negative) for a buyer if someone else acquires an item.

\cite{haghpanah2013optimal} consider the settings where externalities depends on the social network among buyers.

Externalities and network effects \cite{choi2010network}.}

\paragraph{Auction Design with Externalities}
There is ample prior work studying auction design with network externalities in the following sense: if the item for sale is a phone, then one consumer's value for the phone increases when another consumer purchases a phone as well (which is a positive externality, because they can talk to more people). Similarly, the item could be advertising space, in which case one consumer's value for advertising space could decrease as other consumers receive space (which is a negative externality, as now each unit of space is less likely to grab attention)~\cite{haghpanah2013optimal, jehiel1996not, mirrokni2012fixed, candogan2012optimal, bhattacharya2011allocations, hartline2008optimal}. Our work differs in that it is a third party, who is neither selling nor purchasing an item, who suffers the externalities. 

\paragraph{Improving the Commons} There is also a large body of work studying the regulation of common goods (e.g., clean air, security, spectrum access) in the form of taxes or licenses. For example, a government agency can regulate the emission of pollution by auctioning licenses (perhaps towards minimizing the total social cost---regulation cost plus negative externalities)~\cite{montero2008simple, martimort2016mechanism, seabright1993managing, lehr2005managing, feldman2013pricing, weitzman1974prices}. Our work differs in that our regulations are constrained to guarantee minimum profit to the seller, rather than focusing exclusively on the social good. 

\paragraph{Approximation in Auction Design}
Owing to the inherent complexity of optimal auctions for most settings of interest, it is now commonplace in the Economics and Computation community to design simple but approximately auctions. Our work too follows this paradigm. We refer the reader to previous work~\cite{hartline2013mechanism} for an overview of this literature.

\paragraph{Mitigating Security Problems}

Computer security is a particular example of the Tragedy of the Commons, where a software or hardware provider sells an insecure product, and where consumers may purchase the product without considering or taking actions to reduce the security risks. In addition, users might be unable to distinguish insecure products from insecure ones \cite{akerlof1978market}. One mitigation strategy is to have the vendor release updates with security features, although this could be a costly process, as August observes~\cite{august2016market}. However, identifying the existence of security vulnerabilities in the first place may take time for the vendors; for instance, a common software vulnerability known as buffer overflow remained in more than 800 open-source products for a median period of two years before the vendors fixed the problems, according to a study by Li~\cite{li2017large}.

An alternative to relying on a vendor to implement security features or releasing updates is to incentivize the users to follow security practices. Redmiles has found that users who adopt security practices, like using two-factor authentication, have a lower overall utility for themselves than if they adopt no security practices at all, as security practices may introduce inconvenience~\cite{dancingpigs}. Even if users were notified of security problems that they were presumably unaware of, it took as long as two weeks for fewer than 40\% of the users to take remedial actions, according to a study~\cite{li2016you}. To introduce incentives, vendors could, for instance, offer discounts to users who adopt security behaviors~\cite{august2016market}; regulators, on the other hand, could incur fines to users whose software or devices were hacked~\cite{kunreuther2003interdependent}, which is a part of our model in this paper.
\section{Model}\label{sec:model}
In this section, we introduce our model, which consists of a population of rational buyers and a single item for sale. After introducing each of the concepts one-by-one, we include a table (Table~\ref{tab:model}) at the end of this section to remind the reader of each of the components.

\paragraph{Buyer properties} Buyers in our model have two parameters: $(v,k) \in \mathbb{R}^2_+$. $v$ denotes the buyer's value for the item (i.e., how much value does the buyer derive from the IoT device in isolation, independent of fines, etc.). $k$ denotes the buyer's \emph{effectiveness} in translating effort into improved security. That is, a buyer with high $k$ can spend little effort and greatly reduce the risk of being hacked (e.g. because they are well-versed in security measures). A buyer with low $k$ requires significant effort for minimal security gains. We will often use $t:= (v,k)$ to denote a buyer's \emph{type}. 

\paragraph{Security} A buyer who chooses to purchase an item will spend some level of effort $h\geq 0$ securing it, which causes disutility $h$ to the buyer. The seller may also include some default security level $c$. If the buyer has effectiveness $k$, we then denote the combined effort by $\effort(k, c, h):= c+kh$. The idea is that buyers with higher effectiveness are more effective at securing the device for the same disutility. Note that buyers with effectiveness $k > 1$ are more effective than the producer, and buyers with effectiveness $k < 1$ are less effective. Highly effective buyers should not necessarily be interpreted as ``more skilled'' than producers, but some security measures (e.g., password management) are simply more effective for consumers than producers to implement. 

We model the probability that a device is compromised as a function $g(\cdot)$ of \effort, with $g(x):=e^{-x}$. This modeling decision is clearly stylized, and meant as an approximation to practice which captures the following two important features: (a)~as effort $x$ approaches $\infty$, $g(x) \rightarrow 0$ (that is, it is possible to shrink the probability of being compromised arbitrarily small with sufficient effort), and (b)~$g''(x)\geq 0$. That is, the initial units of effort are more effective (i.e. $g'(x)$ is larger in absolute value) than latter ones (when $g'(x)$ is smaller in absolute value). The idea is that consumers/producers will take the highest ``bang-for-buck'' steps first (e.g., setting a password). Note that our results do not qualitively change if, for instance, $g(x):= \lambda_1 e^{\lambda_2 x}$ for some constants $\lambda_1\in (0,1],\lambda_2>0$, but since the model is stylized anyway we set $\lambda_1 = \lambda_2 = 1$ for simplicity of notation. 

\paragraph{Regulatory Policy} The regulator selects a policy/strategy $s = (y, c, p) \in \mathcal \in \mathbb{R}_+^3$. Here, $c$ denotes the security standards the producer must include which is equivalent to the production cost. $y$ denotes the fine the consumer pays should their device be compromised. $p$ denotes the price of the item. Conceptually, one should think of the regulator inducing the producer to set security standard $c$ and price $p$ via particular regulatory policies (e.g., requiring a minimum security level $c'$, or mandating purchase of insurance). Mathematically, we will not belabor exactly how the regulator arrives at $(y,c,p)$. We will also be interested in ``simple'' policies, which regulate either $y$ or $c$.

\begin{definition}[Simple Policy]
For a policy $s = (y, c, p)$, we say $s$ is a \emph{fine policy} if $c = 0$, a \emph{cost policy} if $y = 0$ and a \emph{simple policy} if $s$ is either a \emph{fine policy} or a \emph{cost policy}.
\end{definition}

\paragraph{Utilities} Recall that so far our buyer has value $v$ and efficiency $k$, and chooses to put in effort $h$. The regulator mandates security $c$ (which is equivalent to the production cost) on the item (which has price $p$) and imposes fine $y$ for compromised items. The probability that an item is compromised is $g(\effort(k,c,h))=e^{-c-kh}$. The buyer's utility is therefore: $v - p - h - y\cdot e^{-c-kh}$. Observe that the buyer is in control of $h$ (but not $v, p, y, c, k$). So the buyer will optimize over $h\geq 0$ to minimize $h + y \cdot e^{-c-kh}$. By taking the derivative with respect to $h$, we get a closed form for the choice of effort $h^*(t, s)$ (recalling that we denote the buyer's type $t=(v,k)$ and the regulator's strategy $s=(y,c,p)$):
\begin{equation}\label{eq:buyer-effort}
h^*(t, s)= \max \bigg (0, \frac{\ln(yk) -c}{k} \bigg )
\end{equation}

We can now see that the probability that the buyer's item is compromised, conditioned on expending the optimally chosen effort is:
\begin{equation}\label{eq:risk}
\risk(t, s):= \min\bigg\{e^{-c},\frac{1}{yk}\bigg\}.
\end{equation}

We will additionally refer to the buyer's (security) \emph{loss} as the expected fines they suffer plus the effort they spend. That is:
\begin{equation}\label{eq:loss} \ell(t, s):= y \cdot \risk(t, s)+h^*(t, s):=\begin{cases}
\frac{\ln(yk)-c+1}{k}, \quad &yk \geq e^c\\
ye^{-c}, \quad &yk < e^c
\end{cases}\end{equation}

It then follows that the buyer's utility (value minus price minus expected fines) is:
\begin{equation}\label{eq:buyer-utility}
u(t, s):= v - p - \ell(t, s) = \begin{cases} v - p - 1/k - \frac{\ln(yk)-c}{k}, &yk \geq e^c\\
v-p-ye^{-c}, &yk < e^c
\end{cases}
\end{equation}

\paragraph{Population of Buyers} We model the population of buyers as a distribution $D$ over types $t$. Additionally, we make the now-typical assumption in the multi-dimensional mechanism design literature (e.g.~\cite[and follow-up work]{ChawlaHK07, HartN17}) that the parameters $v$ and $k$ are drawn independently, so that $D:= D_v \times D_k$.\footnote{This assumption is even more justified in our setting than usual, as it is hard to imagine correlation between the \emph{value} a consumer derives from using a smart refridgerator and their \emph{ability} to secure IoT devices.} The seller's profits are then:
\begin{equation}\label{eq:rev}\rev_D(s) := (p-c) \cdot \Pr_{t \leftarrow D}[u(t, s)\geq 0]\end{equation}

\paragraph{Externalities} Finally, we define the externalities caused and the regulator's objective function. Each device sold has some probability of being compromised, and the regulator wishes to minimize the total fraction of compromised devices.\footnote{It would be equally natural for the regulator to aim to minimize the total mass of compromised devices. Most of our results do not rely on optimizing one objective versus the other, but we stick with one in order to unify the presentation.} That is, we measure the externalities caused as:\footnote{Below, $\mathbb{I}(\cdot)$ denotes the indicator function, which takes value $\mathbb{I}(X) = 1$ whenever event $X$ occurs, and $0$ otherwise.}
\begin{equation}\label{eq:ext}\ext_D(s):= \frac{\mathbb{E}_{t \leftarrow D}[\risk(t, s)\cdot \mathbb{I}(u(t, s) \geq 0)]}{\Pr_{t \leftarrow D}[u(t, s) \geq 0]}\end{equation}

\paragraph{Optimization} The regulator's objective is to propose an $s=(y,c,p)$ that minimizes $\ext_D(s)$. Observe that, if left unconstrained, the regulator can simply propose $c \rightarrow \infty$, resulting in $0$ externalities. Such a policy is completely unrealistic, as it would cause costs to approach $\infty$ and destroy the industry. Similarly, taking $y \rightarrow \infty$ would cause consumers to have negative utility even to get the item for free (again destroying the industry). We therefore impose a minimum profit constraint for a policy to be considered feasible. Indeed, this forces the regulator to trade off profits for externalities as effectively as possible. Therefore, our regulator is given some profit constraint $R$, and aims to find:
$${\arg\min}_{s, \rev_D(s) \geq R}\{\ext_D(s)\}$$

We will only consider cases where there is \emph{some} feasible $s$ (that is, we will only consider $R$ such that there exists a $p$ with $R\leq \rev_D(0,0,p)$. If no such $p$ exists, then the profit constraint exceeds the optimal achievable profit without regulation, and the problem is unsolvable). 

\paragraph{Recap of model} Table~\ref{tab:model} recaps the parameters of our model for future reference. Note also that many parameters (e.g. $\ell(\cdot,\cdot)$ are formally defined as a function of $t = (v,k)$ and $s = (y,c,p)$, but only depend on (e.g.) $k,y,c$. As such, it will often be clearer to overload notation and write $\ell(k,y,c)$, rather than defining a new $t=(v,k)$ with a meaningless parameter. Sometimes, though, it will be clearer to use the defined notation for a type $t$ that was just defined. In the interest of clarity, we will overload notation for these variables, but it will be clear from context what they refer to.
\begin{table}[h]
\caption{Model Variables}
\label{tab:model}
\begin{center}
    \begin{tabular}{ || c || c | c | c | c | c ||}
    \hline 
     Variable & Text Definition & Formal Definition \\ \hline 
    $t=(v,k)$ & (value, effectiveness) & N/A\\ 
\hline
$D:= D_v \times D_k$ & buyer population & N/A\\
\hline
$s = (y,c,p)$ & (fine, security, price) & N/A\\
\hline
$h^*(t, s)$ & buyer optimal effort & $\max\{0,\frac{\ln(yk)-c}{k}\}$,\eqref{eq:buyer-effort}\\
\hline
$\risk(t, s)$ & compromise prob. & $\min\{e^{-c},\frac{1}{yk}\}$, \eqref{eq:risk}\\
\hline
$\ell(t, s)$ & buyer security loss
& Equation~\eqref{eq:loss}\\
\hline
$u(t, s)$ & buyer utility & $v-p-\ell(t, s)$, \eqref{eq:buyer-utility} \\
\hline
$\rev_D(s)$ & seller profits & Equation~\eqref{eq:rev}\\
\hline
$\ext_D(s)$ & frac. compromised & Equation~\eqref{eq:ext}\\
\hline

    \end{tabular}
\end{center}
\label{table:prior}
\end{table}

\paragraph{Final Thoughts on Model} We propose a stylized model to capture the following salient aspects of this market: (a) neither buyer nor seller suffer externalities when the item is compromised, (b) the regulator can regulate both the product (via $c$) and payments (via $y,p$), (c) there is a population of buyers, each with different value $v$ and effectiveness $k$ at translating effort into security, and (d) the regulator must effectively trade off externalities with profits by minimizing negative externalities, subject to a minimum profit constraint $R$. The goal of this model is not to capture every potentially relevant parameter, but to isolate the salient features above.

\begin{figure*}[t]

\centering
\ifwww
\else
\includegraphics[scale=0.33]{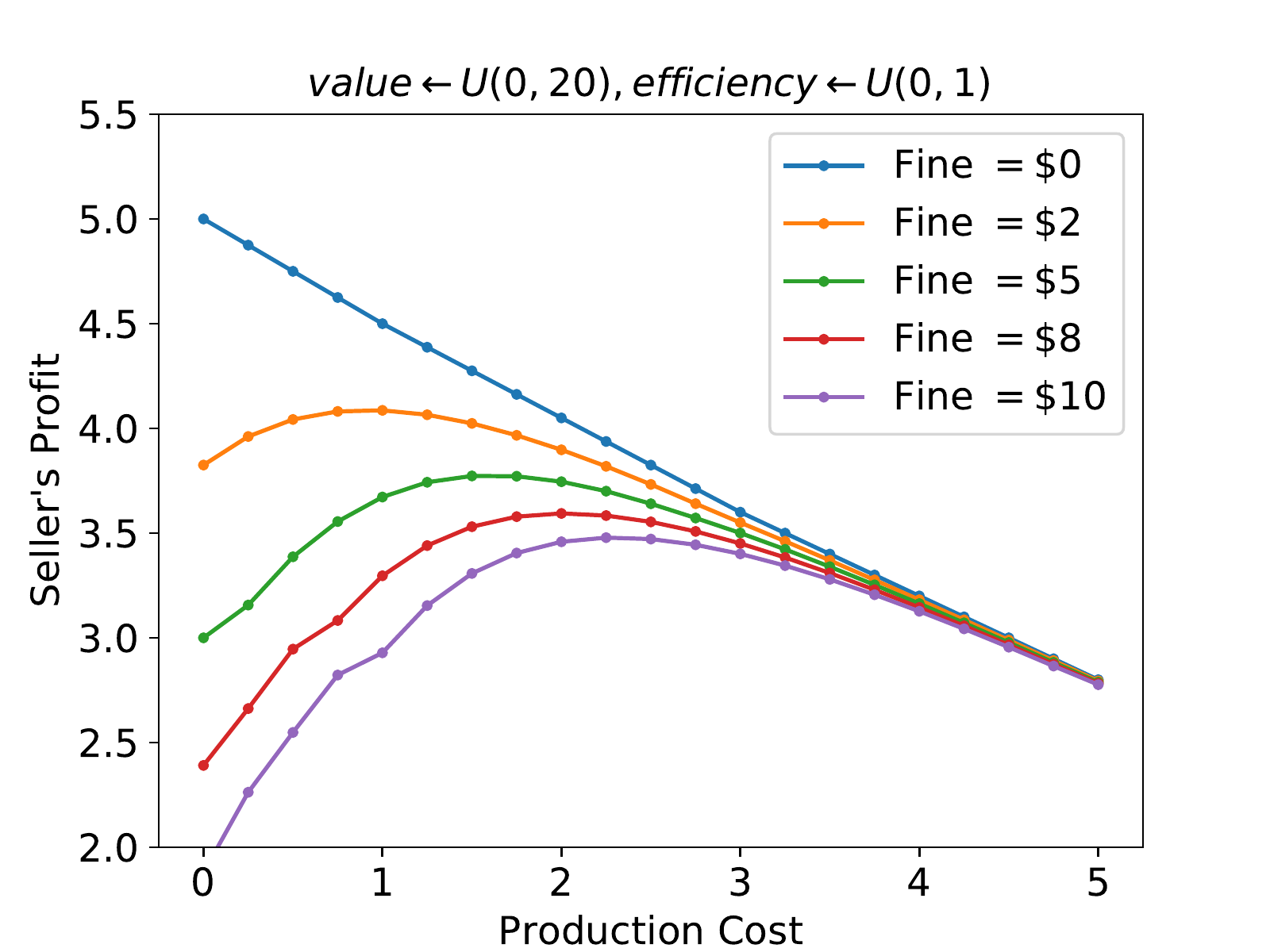}
\includegraphics[scale=0.33]{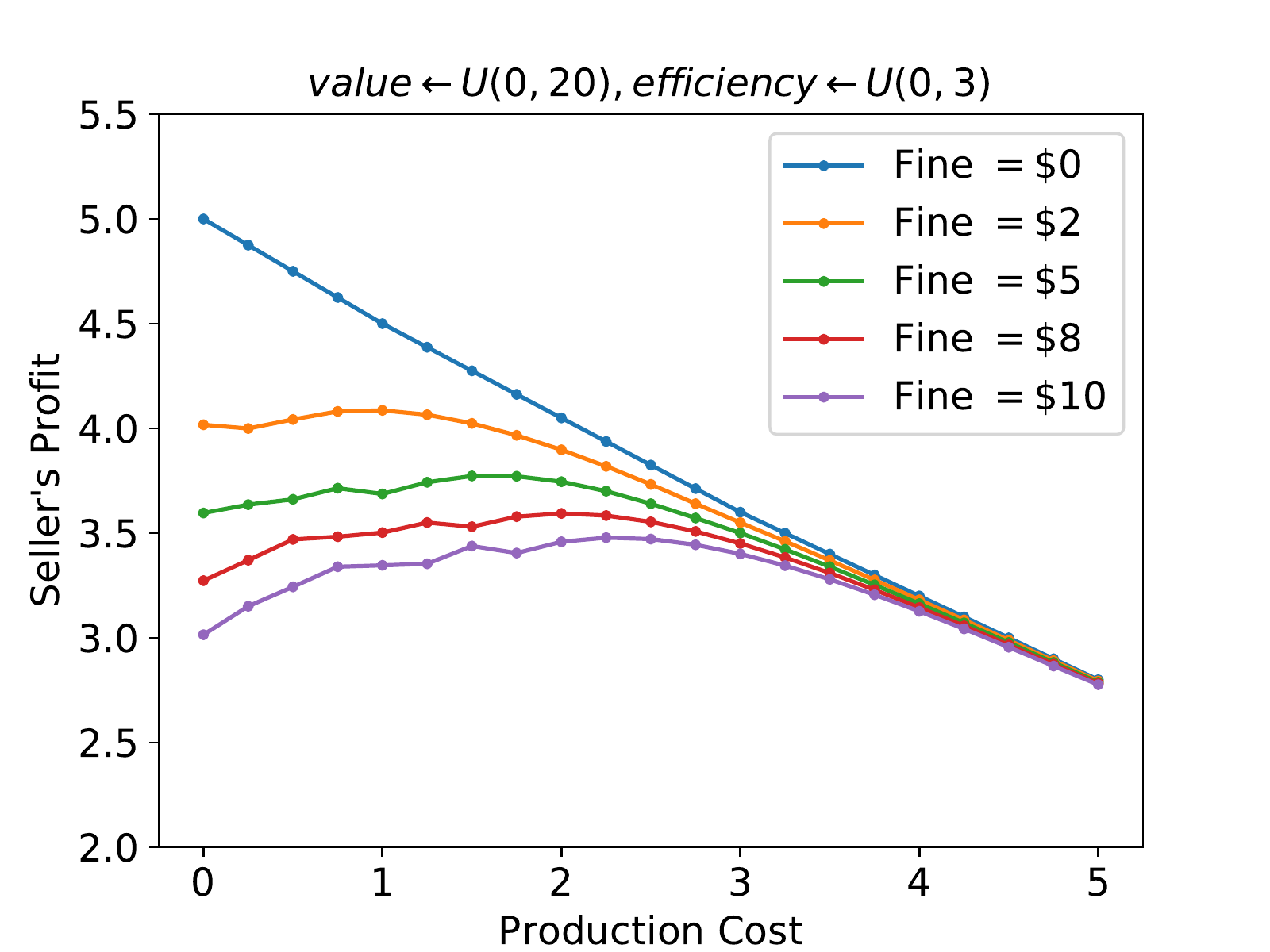}
\includegraphics[scale=0.33]{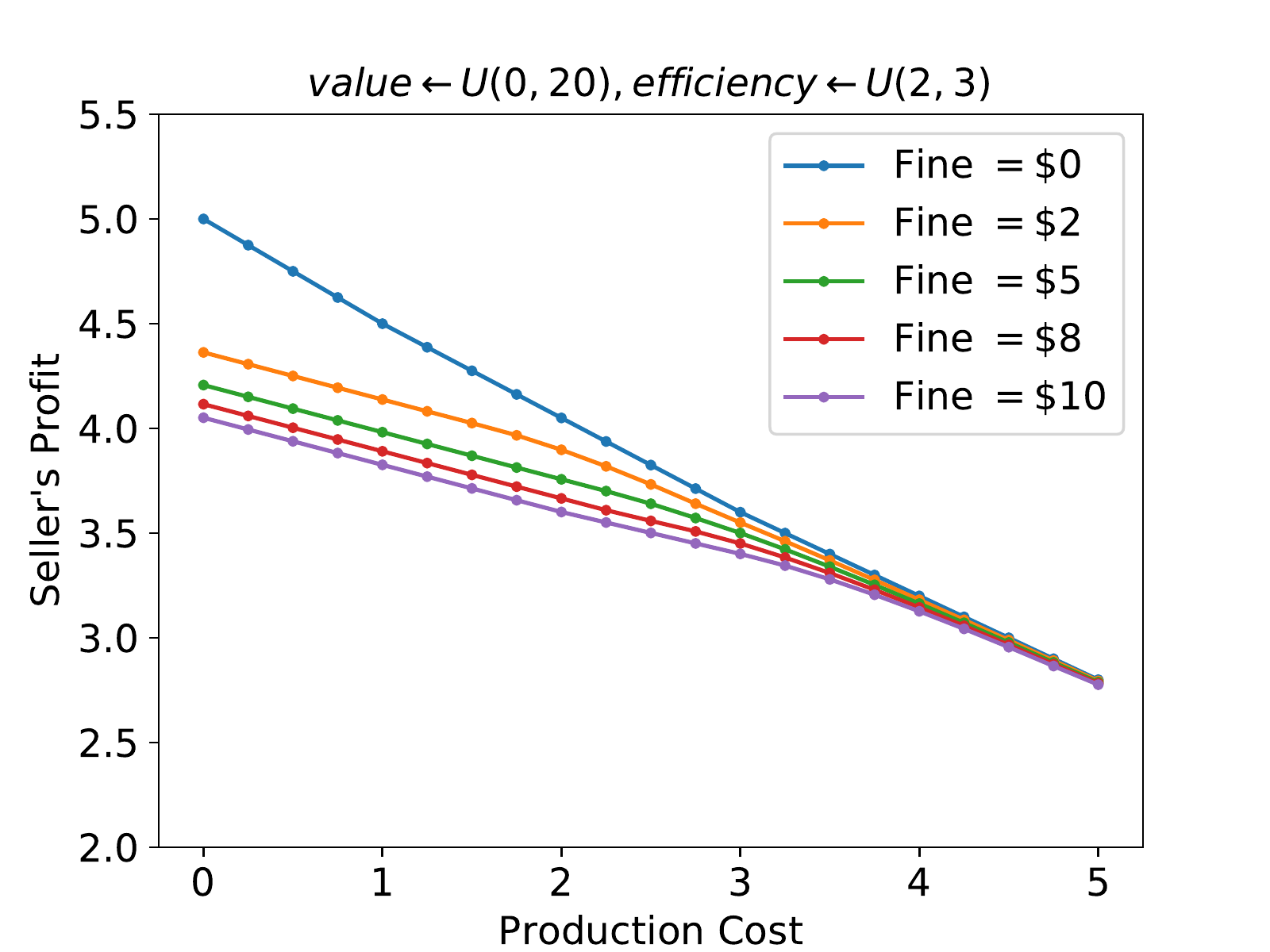}
\fi
\caption{Seller's optimal profits under different distributions for efficiency, $k$. We plot the seller's profits on the vertical axis and the default security $c$ on the horizontal axis. Each curve corresponds to different fines, $y$. Importantly, observe that when the fine is zero, the seller achieves greatest profits with lower default security. However, when the fine is non-zero, the seller may actually \emph{increase} their profits with default security, but the benefits (to the seller) of default security decrease as the buyer population becomes more efficient.}
\label{fig-profit}
\end{figure*}

\subsection{An Intuitive Example}\label{sec-model-example}

In this section, we provide one example to help give intuition for the interaction between the fines $y$, default security $c$, and seller's profits $\rev_D(s)$. In particular, Figure~\ref{fig-profit} plots the maximum achievable $\rev_D(s)$ over all $s$ with a fixed $c$ (the $x$-axis) and $y$ (the color of the plot). In all three examples, $D_v$ is the uniform distribution on $[0,20]$, and $k$ is drawn from either the uniform distribution on $[0,1], [0,3]$, or $[2,3]$, respectively. Note that $k \geq 1$ is the threshold when a buyer is more efficient than the seller in mitigating externalities, so these examples cover two \emph{homogeneous} populations, where all consumers are more (respectively, less) efficient than the producer, and one \emph{heterogeneous} population, where some consumers are more efficient, and others are not.

For each possible (partial) regulation $(y, c)$, the profit-maximizing choice of $p$ is essentially a classic single-item problem (e.g.~\cite{myerson1981optimal}), as the buyer's ``modified value'' $v'$ is simply $v - \ell(t,s)-c$, and the seller's profit for setting price $p$ is just $p \cdot \Pr_{t \leftarrow D}[v' \geq p]$. Therefore, for each partial regulation $(y,c)$, we can construct the modified distribution and simply maximize $p \cdot \Pr_{t \leftarrow D}[v' \geq p]$ as above. 

Observe in Figure~\ref{fig-profit}, when $y=0$, the seller gets greater profits with lower $c$. This should be intuitive, as neither the buyer nor seller suffer when the device is compromised. When $y > 0$, and $D_k = U([0,1])$, the seller's profits can increase with $c$. This should also be intuitive: now that the buyer suffers when the device is compromised, they prefer to buy a secure device. 

On the other hand, when the market contains only efficient buyers ($k > 1$ always), the buyer prefers to provide her own security; any increased cost will always decrease the buyer's utility. Indeed, observe that $\frac{\partial \ell(t,s)}{\partial{c}}$ is either $0$ (if $yk < e^c$) or $1-1/k$ (otherwise). If $k > 1$, then this is always positive, so higher $c$ results in (weakly) higher loss for the consumer, and lower utility.

\arxiv{
\subsection{Example of Efficiency Distribution}

In our model, we will not make any assumptions on the efficiency distribution $D_k$, but we provide an example of how one could model such distributions. To construct $D_k$, we can isolate the different features (e.g. encryption and security practices) that affects security of a population and how they combined affects the buyer's effectiveness in providing security.

Consider the case where some IoT devices, such as security cameras, allow the use of two-factor authentication, where the first factor is password-based authentication, and the second factor is based on, for instance, SMS. A user has the choice of using passwords alone for authentication or using the two factors. For users that use passwords alone, the efficiency may depend on the strength of their passwords or how likely the passwords are re-used.  Figure 2(a) illustrates an example where buyers are generally more efficient than sellers if the buyers pick strong passwords. 

For two-factor authentication, the efficiency may depend on the robustness of the second factor; in particular, an SMS-based second factor might be more prone to compromise than hardware-token-based solutions, as SMS messages could be intercepted.\footnote{See \url{https://www.theverge.com/2017/9/18/16328172/sms-two-factor-authentication-hack-password-bitcoin}} Depending on which second factor is implemented by the seller, $k$'s density varies, an example of which is shown in Figure 2(b). In general, however, because the seller has the control over which second factor to use, the seller is more efficient than the buyer. In Figure~\ref{fig:efficiency} (c), we use mixture distribution to model $D_k$ when password authentication is combined with SMS authentication. In general, systems with two factor authentication allow users to reset passwords ($1^{st}$-factor) through SMS ($2^{nd}$-factor) which suggests the second factor carries higher weight than passwords in the buyer's efficiency. We define $Pr[D_k = x] = \frac{2}{3}Pr[D_1 = x] + \frac{1}{3}Pr[D_2 = x]$ where the weights model the fact the second $2^{nd}$ factor can override the $1^{st}$ factor even though SMS authentication can be vulnerable.

\begin{figure*}[t]
\centering
\includegraphics[scale=0.55]{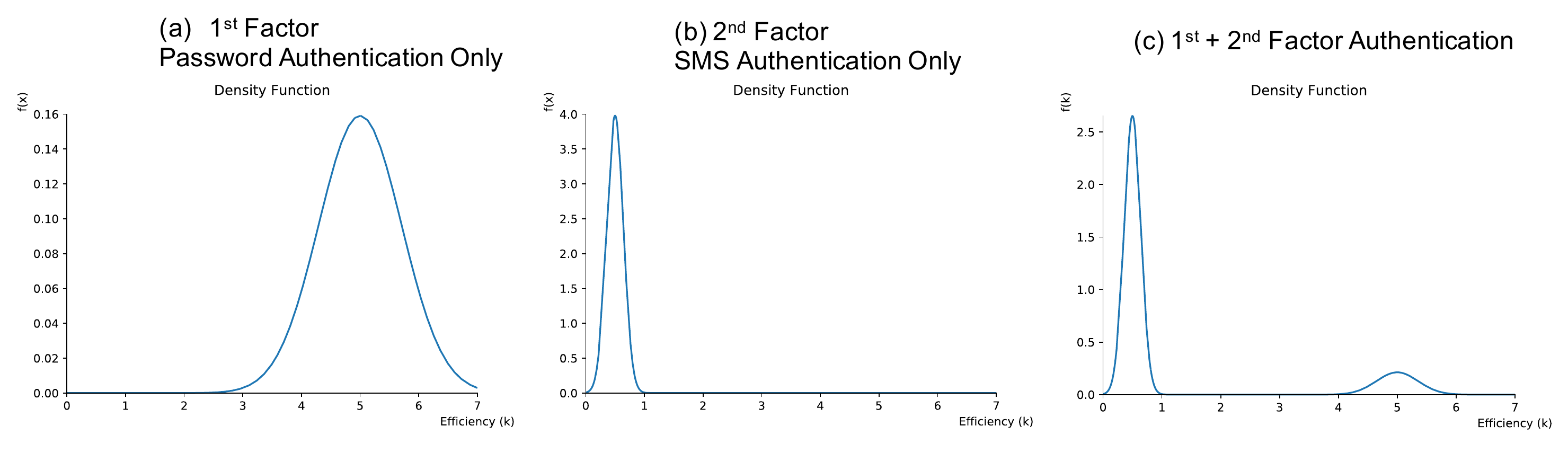}
\caption{The efficiency distribution of a population. We plot the effectiveness of a buyer when restricted only to password authentication (left) and SMS authentication (center). By combining the features, and taking into account their contributions in minimizing externalities, we construct the efficiency distribution in the right.}
\label{fig:efficiency}
\end{figure*}
}
\subsection{Preliminary Observations}\label{sec-preliminaries}
We conclude with two observations which allow an easy comparison between the profits of certain policies. Intuitively, Observation~\ref{obs:monotone} claims that any policy which makes \emph{every single} consumer in the population have lower loss generates greater profits for the seller. We will make use of Observation~\ref{obs:monotone} repeatedly throughout the technical sections to modify existing policies into ones which improve profits (ideally while also improving externalities, although that is not covered by Observation~\ref{obs:monotone}). 

\begin{observation}\label{obs:monotone}
Let $s=(y,c,p)$, $s'=(y', c',p')$ be such that $p' -c' = p-c\geq 0$ and for all $k \in \text{support}(D_k)$, $\ell(k,s)+c \leq \ell(k,s')+c'$. Then for all $D_v$, $\rev_{D_v \times D_k}(s) \geq \rev_{D_v \times D_k}(s')$. 
\end{observation}
\begin{proof}
Observe that for both $s$ and $s'$, the seller's profit per sale is identical (as $p' - c' = p-c$). So we just wish to show that the probability of sale for $s'$ is larger than that for $s$. Indeed, observe that for all $t$:
\begin{align*}
    u(t,s)&=v-\ell(k,s)-p\\
    &=v-(\ell(k,s)+c)+c-p\\
    &\geq v - (\ell(k,s')+c') +c'-p'\\
    &= v-\ell(k,s')-p' = u(t,s').
\end{align*}
Therefore, any consumer $(v,k)$ who chooses to purchase the item under policy $s'$ will also choose to purchase under policy $s$, and therefore the probability of sale is at least as large for $s$ as $s'$. 
\end{proof}

Observation~\ref{obs:monotone2} below claims that the profit of any policy $s$ is larger in populations $D$ where every consumer is more effective than in $D'$.

\begin{observation}\label{obs:monotone2}
Let $D_k$ \emph{stochastically dominate} $D'_k$.\footnote{That is, it is possible to couple draws $(k,k')$ from $(D_k, D'_k)$ so that $k \geq k'$ with probability $1$. Equivalently: for all $x$, $\Pr[k \geq x, k \leftarrow D_k] \geq \Pr[k' \geq x, k' \leftarrow D'_k]$.} Then for all policies $s=(y,c,p)$ with $p \geq c$, and all $D_v$, $\rev_{D_v \times D_k}(s) \geq \rev_{D_v \times D'_k}(s)$. 
\end{observation}
\begin{proof}
As $D_k$ stochastically dominates $D'_k$, it is possible to couple draws $(t, t')$ from $(D_v \times D_k, D_v \times D'_k)$ such that $v = v'$ and $k \geq k'$. Observe simply that $u(t,s)\geq u(t',s)$ always. Therefore, $\Pr[u(t,s) \geq 0] \geq \Pr[u(t',s)\geq 0]$, and $\rev_{D_v \times D_k}(s) \geq \rev_{D \times D'_k}(s)$. 
\end{proof}

Observe, however, that Observation~\ref{obs:monotone2}, perhaps counterintuitively, does \emph{not} hold if we replace profits with externalities. That is, for a fixed policy $s$, we might \emph{increase} all consumers' effectiveness yet also \emph{increase} the externalities caused. Intuitively, this might happen (for instance) in a fine policy which successfully only sells the item to extremely effective consumers who effectively secure their purchase. Ineffective consumers choose not to purchase the product to avoid fines. However, if these ineffective consumers are instead somewhat effective, they may now choose to purchase the item, thereby increasing externalities. Below is a concrete instantiation:

\begin{example}\label{ex:ext-non-monotonicity}
Consider the population where $D_v$ is a point-mass at $e$, and $D_k$ takes on effectiveness $0$ with probability $1/2$ and $x> 1$ with probability $1/2$. Consider the policy $s=(e,0,e-2.5)$. Then the $(e,0)$ consumer chooses not to purchase: $\ell(e,0,s) = e$, so their utility would be $e-e-(e-2.5) < 0$. The $(e,x)$ consumer chooses to purchase, as their loss is $\frac{2+\ln(x)}{x} < 2$ (as $x > 1$). So $\ext_{D_v \times D_k}(s) = \frac{1}{ex}$. 

Consider now improving the effectiveness of the $k=0$ consumers to $k=1$ (so $D'_k$ now takes on $1$ with probability $1/2$ and $x$ with probability $1/2$). The $(e,1)$ consumer now chooses to purchase, as their loss is $2$ (so their utility is $e-2-(e-2.5) = 1/2$). So now $\ext_{D_v \times D'_k}(s) = (\frac{1}{e} + \frac{1}{ex})/2$. As $x > 1$, the externalities have gone up. If $x\geq 1$, the externalities may have gone up quite significantly.
\end{example}

In Example~\ref{ex:ext-non-monotonicity}, of course ``the right'' thing to do is to also change the policy. Indeed, it is still the case that, for a \emph{fixed consumer who purchases the item}, increasing effectiveness can only decrease externalities. But without fixing whether the consumer has purchased the item or not, the claim is false. Observation~\ref{obs:monotone3} captures what we can claim about risk, loss, etc. on a per-consumer basis. Proofs for the claims in Observation~\ref{obs:monotone3} all follow immediately from the definitions in Section~\ref{sec:model}.

\begin{observation}\label{obs:monotone3}
Let $k > k'$, then for all $s$:
\begin{itemize}
    \item $\risk(k,s) \leq \risk(k',s)$.
    \item $\ell(k,s) \leq \ell(k',s)$.
    \item $h^*(k,s) \geq h^*(k',s)$.
    \item $u(k,s) \geq u(k',s)$. 
\end{itemize}
\end{observation}
\section{Roadmap of Technical Sections}\label{sec:roadmap}
Now that we have the appropriate technical language, we provide a brief roadmap of the results to come.
\begin{itemize}
    \item In Section~\ref{sec:deterministic-distribution}, we provide a technical warmup to get the reader familiar with how to reason about our problem. The main result of this section is Theorem~\ref{thm:opt-deterministic}, which claims that the optimal policy when $D_k$ is a point-mass is simple. The proof of this theorem helps illustrate one key aspect of our later arguments, and will also be used as a building block for later proofs.
    \item In Section~\ref{sec:homogeneous-distribution}, we prove our first main result (Theorem~\ref{thm:opt-homogeneous}): as a function of $R$ and $D_v$, there exists a cutoff $T$. If $D_k$ is supported on $[0,T]$, then a cost policy is optimal. If $D_k$ is supported on $[T,\infty)$, then a fine policy outperforms all \emph{profits-maximizing} policies (we define this term in the relevant section --- intuitively a policy is profits-maximizing if the price is the seller's best response to $(y,c)$). Section~\ref{sec:homogeneous-distribution} also contains a surprising example witnessing that the additional profits-maximizing qualification is necessary.
    \item In Section~\ref{sec:general}, we consider general distributions. Unsurprisingly, simple policies are no longer optimal. Perhaps surprisingly, if one insists on exceeding the profits benchmark \emph{exactly}, no simple policy can guarantee any bounded approximation to the optimal externalities (Corollary~\ref{cor:lb}). However, we also show (Theorem~\ref{thm:approx}) that it is possible to get a bicriterion approximation: if one is willing to approximately satisfy the profits constraint, it is possible to approximately minimize externalities with a simple policy. That is, for any $s,D$, there is a simple policy $s'$ with $\rev_D(s') =\Omega(1) \cdot \rev_D(s)$ and $\ext_D(s') = O(1) \cdot \ext_D(s)$. 
    \item We include complete proofs for \www{most of} our results on point-mass and homogeneous distributions, as these convey many of the key ideas. \arxiv{By Theorem~\ref{thm:approx}, the proofs get quite technical so we defer them to the appendix.}\www{By Theorem~\ref{thm:approx}, the proofs get quite technical so we provide a sketch of the main ideas. This and other omitted proofs can be found in \cite{venturyne2019externality}.}
\end{itemize}
\section{Warm-up: Point-Mass Effectiveness}\label{sec:deterministic-distribution}

As a warm-up, we first study the case where $D_k$ is a point mass (that is, all buyers in the population have the same effectiveness $k$). In this case, we show that a simple policy is optimal. The proof is fairly intuitive, with one catch. The intuitive part is that every consumer will put in the same effort, conditioned on buying the item. It therefore seems intuitive that if $k <1$, it is better for all parties involved if any effort spent by the consumer is transferred to the producer instead (and this is true). It also seems intuitive that if $k > 1$, it is again better for all parties involved if any effort spent by the producer is ``transferred'' to the consumer instead (e.g. by raising fines so that the consumer chooses to spend the desired level of effort). This is not quite true: the catch is that the fine required to induce the desired buyer behavior may be too high to satisfy the profit constraint. But, the above argument does work for sufficiently large $k$. Importantly, there is some cutoff $T$ such that for all $k \leq T$, the optimal policy is a cost policy ($y=0$), while for all $k \geq T$, the optimal policy is a fine policy ($c=0$). Below, when we write $D_v \times \{k\}$, we mean the distribution which draws $v$ from $D_v$ and outputs $(v,k)$.  

\begin{theorem}
\label{thm:opt-deterministic} For all $D_v$, $R$, and $k$, the externality-minimizing policy for $D_v \times \{k\}$ is a simple policy. Moreover, for all $R, D_v$, there is a cutoff $T$ such that if $k \leq T$, then the optimal policy is a cost policy. If $k \geq T$, then the optimal policy is a fine policy.
\end{theorem}

\begin{proof}
Consider any policy $s=(y,c,p)$. Because all consumers have the same effectiveness $k$, $s$ induces the same loss for all consumers. We first claim the following:

\begin{lemma}\label{lem:easycase}Let $k \leq 1$. Then for all $D_v$ and any policy $s=(y,c,p)$, there is an alternative policy $s'=(0,c',p)$ with $\rev_{D_v\times\{k\}}(s') \geq \rev_{D_v \times \{k\}}(s)$ and $\ext_{D_v \times \{k\}}(s') \leq \ext_{D_v \times\{k\}}(s)$.
\end{lemma}
\begin{proof}
In policy $s$, all consumers have the same loss $\ell(k,s)$. This therefore is a good opportunity to try and make use of Observation~\ref{obs:monotone}. First, consider the possibility that $yk < e^c$. In this case, $h^*(k,s) = 0$, $\ell(k,s) = ye^{-c}$, and $\risk(k,s) = e^{-c}$. This implies that $\ext_{D_v \times \{k\}}(s) = e^{-c}$. Consider instead the policy $s'=(0,c,p)$. Then $\ell(k,s') = 0$, but $\risk(k,s) = e^{-c}$ and $\ext_{D_v \times \{k\}}(s) = e^{-c}$ like before. So the externalities are the same. An application of Observation~\ref{obs:monotone} concludes that the profits have improved (indeed, $(c,p)$ are the same in both policies, and the loss decreases as we switch from policy $s$ to $s'$). 

Consider now the possibility that $yk \geq e^c$. In this case, $h^*(k,s) = \frac{\ln(yk)-c}{k}$, $\ell(k,s)=\frac{\ln(yk)-c+1}{k}$, and $\risk(k,s) = \frac{1}{yk}$. Consider instead the policy $s' = (0,\ln(yk), p-c+\ln(yk))$. In this new policy, $\ell(k,s') = 0$ and $\risk(k,s') = \frac{1}{yk}$. So indeed, the new policy has the same externalities. We just need to ensure that we can apply Observation~\ref{obs:monotone}. To this end, observe that:
\begin{align*}
    \ell(k,s)+c-(\ell(k,s')+c') &= \frac{\ln(yk)-c+1}{k}+c - \ln(yk)\\
    &=(1/k-1)\cdot(\ln(yk)-c) + 1/k\\
    &\geq 0.
\end{align*}
The last line follows because $k\leq 1$ and $\ln(yk) \geq c$ (because $yk \geq e^c$). So the hypotheses of Observation~\ref{obs:monotone} hold, and we can apply Observation~\ref{obs:monotone} to conclude that the profits improve from $s$ to $s'$ as well.
\end{proof}
Lemma~\ref{lem:easycase} covers the cases when $k \leq 1$: there is always an optimal cost policy. We now move to the case when $k > 1$. There are two cases to consider: one where the optimal policy will be a cost policy, and one where the optimal policy will be a fine policy. The distinguishing feature between these cases will be for a given $c$, how big of a fine is necessary to incentivize the consumer to put in effort $c/k$, and what the consumer's loss looks like for this choice of $y$. Below, $c^*$ is defined to be the maximum $c$ such that there exists a $p$ such that $\rev_{D_v \times \{0\}}(0,c,p) \geq R$. Observe that $c^*$ is also equal to the maximum $\ell$ such that there exists a $p$ such that $\rev_{(D_v - \ell) \times \{0\}}(0,0,p) \geq R$ (here, $D_v - \ell$ denotes the distribution which samples $v$ from $D_v$ and then subtracts $\ell$, taking a maximum with $0$ if desired). That is, $c^*$ is the maximum loss that can be uniformly applied to all consumers (drawn from $D_v$) while still resulting in a distribution for which profit $\geq R$ is achievable. 

\begin{lemma}\label{lem:hardercase} Let $c^*$ denote the maximum $c$ such that there exists a $p$ such that $\rev_{D_v \times \{0\}}(0,c,p) \geq R$. Then a cost policy is optimal for $D_v \times \{k\}$ if $k \in [1,1+1/c^*]$. 
\end{lemma}
\begin{proof}
First, observe that the lemma hypothesis implies that \emph{any} feasible policy must have $\ell(k,s) +c \leq c^*$ (if not, then an application of Observation~\ref{obs:monotone} lets us contradict the lemma's hypothesis with a feasible $c'=\ell(k,s)+c>c^*$). 

Consider now $k \in [1,1+1/c^*]$, and start from some policy $s=(y,c,p)$. If this policy has $h^*(k,s) = 0$, then certainly we can just update $s'=(0,c,p)$ and get better profits with the same externalities (by Observation~\ref{obs:monotone}). If instead $h^*(k,s)>0$, then $\ell(k,s)=\frac{\ln(yk)-c+1}{k}$, and $\risk(k,s) = \frac{1}{yk}$. Consider instead $s^* = (0,c^*,p^*)$, for whichever $p^*$ witnesses $\rev_D(s^*) \geq R$ (we know that such a $p^*$ exists by the lemma's hypothesis). So now we just need to compare externalities. Assume for contradiction that $\risk(k,s^*) > \risk(k,s)$. Then we get:
\begin{align*}
    \risk(k,s^*) > \risk(k,s)  &\Rightarrow e^{-c^*} > \frac{1}{yk} \\
    &\Rightarrow c^* < \ln(yk)\\
    &\Rightarrow \frac{\ln(yk)-c+1}{k} >\frac{c^*-c+1}{k}\\
    &\Rightarrow \ell(k,s)+c > \frac{c^*-c+1}{k}+c\\
    &\Rightarrow \ell(k,s)+c > \frac{c^*+1}{k}\\
    &\Rightarrow \ell(k,s)+c > c^* \Rightarrow\!\Leftarrow.
\end{align*}
The last implication uses the fact that $k \leq 1+1/c^*$. The line before this uses that $k \geq 1$. The contradiction arises because this would imply a scheme ($s$) with profit $\geq R$ with loss $>c^*$, contradicting the definition of $c^*$ by the reasoning in the first paragraph of this proof.
\end{proof}

\begin{lemma}\label{lem:hardestcase} Let $c^*$ denote the maximum $c$ such that there exists a $p$ such that $\rev_{D_v \times \{0\}}(0,c,p) \geq R$. Then a fine policy is optimal for $D_v \times \{k\}$ if $k \geq 1+1/c^*$. 
\end{lemma}
\begin{proof}
Again start from some policy $s = (y,c,p)$, inducing some loss $\ell(k,s)$. First, maybe $h^*(k,s)>0$. In this case, the risk is $\frac{1}{yk}$ and the loss plus cost is $\frac{\ln(yk)-c+1}{k}+c$. In particular, observe that the partial derivative of the loss plus cost with respect to $c$ is $1-1/k > 0$. So the policy $s' = (y,0,p-c)$ has $\risk(k,s') = \risk(k,s)$ but also $\ell(k,s')+c' <\ell(k,s)+c$. So Observation~\ref{obs:monotone} claims that this policy gets at least as much profits (and the risk is the same). 

If instead, $h^*(k,s) = 0$, then the risk is $e^{-c}$ and the loss is $y\cdot e^{-c}$. In this case, consider instead $y^*$ such that $\frac{\ln(y^*k)+1}{k} = c^*$ and using $s^* = (y^*,0,p^*)$, for the $p^*$ satisfying $\rev_D(s^*) \geq R$ (again, such a $p^*$ must exist by definition of $c^*$, and the fact that $\ell(k,s^*) = c^*$, plus Observation~\ref{obs:monotone}). We just need to analyze the risk. Similar to the previous proof, assume for contradiction that $\risk(k,s^*)>\risk(k,s)$. Then:
\begin{align*}
    \risk(k,s^*) > \risk(k,s)  &\Rightarrow e^{-c} < \frac{1}{y^*k} \\
    &\rightarrow c > \ln(y^*k)\\
    &\Rightarrow \frac{\ln(y^*k)+1}{k} <\frac{c+1}{k}\\
    &\Rightarrow c^* < \frac{c+1}{k}\\
    &\Rightarrow c^* > c^* \cdot \frac{c+1}{1+c^*} \Rightarrow \!\Leftarrow.
\end{align*}
The last inequality uses the fact that $k \geq 1+1/c^*$, and derives a contradiction as $c \leq c^*$ (if $c > c^*$, then certainly $\ell(k,s)+c>c^*$, contradicting the definition of $c^*$). 
\end{proof}

All three cases together prove Theorem~\ref{thm:opt-deterministic}. The $T$ prescribed in the theorem statement is exactly $1+1/c^*$, where $c^*$ is the maximum $c$ such that there exists a $p$ for which $\rev_{D_v\times \{0\}}(0,c,p) \geq R$.
\end{proof}

We conclude with one last proposition regarding the behavior of the threshold with respect to the profits constraints $R$. Proposition~\ref{prop:thresholds} below states that as $R$ increases, the threshold beyond which a fine policy is optimal increases as well.

\begin{proposition}\label{prop:thresholds} Let $T(D_v,R)$ denote the threshold such that both a fine policy and cost policy are optimal for $D_v\times \{T(D_v,R)\}$ subject to profits constraints $R$. Then $T(D_v,R)$ is monotone increasing in $R$.
\end{proposition}
\begin{proof}
To see this, let $c^*(D_v,R)$ denote the maximum $c$ such that there exists a $p$ such that $\rev_{D_v \times \{0\}}(0,c,p) \geq R$. Then $c^*(D_v,R)$ is \emph{decreasing} in $R$ (as the profits constraint goes up, we can't afford as much security). So $1+1/c^*(D_v,R)$ is \emph{increasing} in $R$. This means that the threshold $T(D_v,R)$ beyond which a fine policy is optimal for $D_v \times \{T\}$ is increasing as a function of the profits constraint $R$ (because $T = 1+1/c^*(D_v,R)$). 
\end{proof}

This concludes our treatment of the case where $k$ is a point-mass. Theorem~\ref{thm:opt-deterministic} should both be viewed as a warm-up to introduce some of our core techniques, and also as a building block towards our stronger theorems (in the following sections). The main technique we introduced is the ability to reduce risk and loss simultaneously to improve both profits and externalities. The idea was that if the buyer is less effective than the seller, everyone prefers that the seller put in effort ($y=0, c> 0$). If the buyer is more effective than the seller, everyone prefers that the buyer put in effort. However, the regulator can not directly mandate that the buyer put in effort, and unfortunately the fines required to extract the desired buyer behavior may too negatively affect the profit. This is why the transition from cost to fine policies is $1+1/c^*$ instead of $1$. 

\section{Homogeneous Distributions}\label{sec:homogeneous-distribution}
In this section, we show that for populations that are sufficiently homogeneous in effectiveness, the optimal policy remains simple. The second half of Theorem~\ref{thm:opt-homogeneous} requires a technical assumption. Specifically, we say that a policy $(y,c,p)$ is \emph{profits-maximizing} if, conditioned on $y,c$, $p$ is set to maximize the seller's profits (that is, $\rev_D(y,c,p) \geq \rev_D(y,c,p')$ for all $p'$).

\begin{theorem}\label{thm:opt-homogeneous} For all $D_v$, $R$, there exists a cutoff $T$ such that \begin{itemize}
    \item For all $D_k$ supported on $[0,T]$, the externality-minimizing policy for $D_v\times D_k$ subject to profits $R$ is a cost policy.
    \item For all $D_k$ supported on $[T,\infty)$, the externality-minimizing policy for $D_v \times D_k$ subject to profits $R$ is either a fine policy, or it is \emph{not} profits-maximizing.
    \end{itemize}
\end{theorem}

The proof of Theorem~\ref{thm:opt-homogeneous} will follow from Lemmas~\ref{lem:costextension} and~\ref{lemma:fineextension}, which handle the two claims in the theorem separately. Finally, we show in Section~\ref{sec:newexample} that the profits-maximizing qualification in part two of Theorem~\ref{thm:opt-homogeneous} is necessary: 

\begin{proposition}\label{prop:newexample}
There exist distributions $D_v, D_k$, and profits constraint $R$ such that:
\begin{itemize}
\item $T$ is such that for all $k \geq T$, the externality-minimizing policy for $D_v\times\{k\}$ subject to profits constraints $R$ is a fine policy. 
\item $D_k$ is supported on $[T,\infty)$.
\item \emph{No} fine policy is externality-minimizing policy for $D_v \times D_k$ subject to profits constraints $R$.
\item The externality-minimizing policy for $D_v \times D_k$ subject to profits constraints $R$ is not simple, and not profits-maximizing (the latter is implied by the second bullet of Theorem~\ref{thm:opt-homogeneous}). 
\end{itemize}
\end{proposition}

Proposition~\ref{prop:newexample} is perhaps surprising: a fine policy is externality-minimizing for $D_v \times \{T\}$, and $D_k$ stochastically dominates $T$, so the same fine policy has \emph{even lower} externalities, and potentially greater profit for $D_v \times D_k$. Indeed, the optimal fine policy for $D_v \times D_k$ achieves lower externalities than that of $D_v \times \{T\}$. The catch is that an \emph{even better} non-simple policy becomes viable, and achieves \emph{still lower} externalities. Theorem~\ref{thm:opt-homogeneous} claims, however, that the optimal non-simple policy must not be profits-maximizing. \arxiv{Even more surprising, we show in \Cref{sec:profits-maximizing-seller} that if we constrain the optimization problem to only profits-maximizing prices then \Cref{prop:newexample} is still true which implies the negation of the second bullet of \Cref{thm:opt-homogeneous}.}

\subsection{Extension Lemma for small \texorpdfstring{$k$}{k}}\label{sec:extension-cost}
The small $k$ case follows roughly from the following intuition. For cost policies, neither the buyer's loss nor her risk depend on $k$. So whichever cost policy is optimal for $D_v \times \{T\}$ achieves the same profits and externalities as $D_v \times D_k$. Intuitively, going from $\{T\}$ to $D_k$ supported on $[0,T]$ cannot possibly \emph{increase} the profits of any scheme (formally: Observation~\ref{obs:monotone2}), so the initial cost policy should remain optimal. 

\begin{lemma}[Extension of Cost Policy]\label{lem:costextension}
Let $s$ be a cost policy that is optimal for $D_v \times \{T\}$ subject to profits $R$. Then for all $D_k$ supported on $[0,T]$, $s$ is optimal for $D_v \times D_k$ subject to profits $R$. 
\end{lemma}
\begin{proof}
First, we observe that $\rev_{D_v \times \{T\}}(s) = \rev_{D_v \times D_k}(s)$. This is simply because the loss of consumers is independent of $k$ (as $y = 0$). Similarly, $\ext_{D_v \times \{T\}}(s) = \ext_{D_v \times D_k}(s)$. This is again because the risk of consumers is independent of $k$. 

Now, assume for contradiction that there is some policy $s'$ with profits $\rev_{D_v \times D_k}(s') \geq R$ and also $\ext_{D_v \times D_k}(s') < \ext_{D_v \times \{T\}}(s)$. Then we have the following inequality from Observation~\ref{obs:monotone2}:
$$R \leq \rev_{D_v \times D_k}(s') \leq \rev_{D_v \times \{T\}}(s').$$

Therefore, as $s$ is optimal for $D_v \times \{T\}$ subject to profits $R$, we must have:
$$\ext_{D_v \times \{T\}}(s') \geq \ext_{D_v \times \{T\}}(s).$$

This now lets us conclude the following chain of inequalities, where the first line is a corollary of Observation~\ref{obs:monotone3}: the consumer in a population with $D_k$ supported on $[0,T]$ whose device is least likely to be compromised is a consumer with $k = T$. The third line follows from the reasoning above (that $s'$ achieves profits at least $R$ on $D_v \times \{T\}$, and is therefore feasible). The final line follows because the externalities of a cost policy are independent of $k$.
\begin{align*}
    \ext_{D_v \times D_k}(s') &\geq \risk(T,s')\\
    &=\ext_{D_v \times \{T\}}(s')\\
    &\geq \ext_{D_v \times \{T\}}(s)\\
    &= \ext_{D_v \times D_k}(s).
\end{align*}
\end{proof}

Lemma~\ref{lem:costextension} proves the first bullet of Theorem~\ref{thm:opt-homogeneous}.

\subsection{Extension Lemma for large \texorpdfstring{$k$}{k}}\label{sec:extension-fine}

In this section, we \www{sketch the proof}\arxiv{proof} for the large $k$ case of Theorem~\ref{thm:opt-homogeneous}. \www{Refer to \Cref{sec:homogeneous-distribution} of \cite{venturyne2019externality} for omitted proofs.} The proof will be a little more involved this time, since we can no longer claim that the externalities of a fine policy are independent of $k$ (whereas this does hold for cost policies). The intuition for this case is the same though: if a fine policy is optimal for $D_v \times \{k\}$ for all $k \geq T$, and $D_k$ is supported on $[T,\infty)$, fine policies should remain optimal for $D_v \times D_k$. Most of the proof does not make use of the technical assumption that the $s$ we are competing with is a profits-maximizing policy: this assumption only arises at the very end.

The first step in our proof is the following concept, which captures the change in loss for a consumer $(v,k)$ for regulation $s$ versus $s'$:

\begin{definition}[Policy Comparison Function]\label{policy-transfer-function}
For two policies $s$ and $s'$, we define the \emph{policy comparison function} $g_{s,s'}(\cdot)$ so that $g_{s, s'}(k) = \ell(k, s) - \ell(k, s')$. 
\end{definition}

The policy comparison function takes as input an effectiveness $k$, and outputs the change in loss for a consumer under one policy versus another. Our first lemma argues that for certain pairs $(s,s')$, the policy comparison function is monotone in $k$. That is, consumers with more effectives have greater preference for one policy over another.

\begin{lemma}\label{lemma:transfer-function-monotonicity}
Let $s=(y,c,p)$ and $s'=(y',c',p')$ be such that $ye^{-c}\leq y'e^{-c'}$. Then $g_{s, s'}(\cdot)$ is monotone non-decreasing. Observe that the hypothesis holds if $y \leq y'$ and $c \geq c'$. 
\end{lemma}

\saveproof{lemmatransferfunctionmonotonicity}{lemma:transfer-function-monotonicity}{
There are three regions of $k$ to consider: $k \in [0,e^{c'}/y')$, $k \in [e^{c'}/y',e^c/y)$, and $k \geq e^{c}/y$. In the first region the consumer has effort $0$ for both policies. In the middle region, the consumer has effort $0$ for one policy. In the last region, the consumer has non-zero effort for both policies.

First consider $k \leq e^{c'}/y' \leq e^c/y$. Then also $y'k\leq e^{c'}$ and $yk\leq e^c$, so the consumer's effort is $0$. In this case, $\frac{\partial \ell(k,s)}{\partial k} = 0=\frac{\partial \ell(k,s')}{\partial k}$, because the loss is independent of $k$. So $g_{s,s'}$ is monotone non-decreasing in this range (in fact it is constant).

Next, there are $k$ such that $yk < e^c, y'k \geq e^{c'}$. Then $\frac{\partial \ell(k,s)}{\partial k} = 0$, and $\frac{\partial \ell(k,s')}{\partial k} = -\frac{\ln(y'k)-c'}{k^2} \leq 0$. So $g_{s,s'}$ is also monotone non-decreasing in this range (because it is equal to $0$ minus a non-increasing function). 

Finally, there are $k$ such that $yk \geq e^c$. Then $\frac{\partial \ell(k,s)}{\partial k} = -\frac{\ln(yk)-c}{k^2}$ and $\frac{\partial \ell(k,s)}{\partial k} = -\frac{\ln(y'k)-c'}{k^2}$. And we get:
\begin{align*}
    \frac{\partial \ell(k,s)}{\partial k}-\frac{\partial \ell(k,s)}{\partial k} &= \frac{\ln(y'k) - c' - \ln(yk) + c}{k^2}\\
    &=\frac{\ln(y'/y) +c-c'}{k^2} \geq 0.\\
\end{align*}
The final inequality comes because by hypothesis: $y'e^{-c'}\geq ye^{-c} \Rightarrow \ln(y')-c' \geq \ln(y) -c$. So in all regions, $g_{s,s'}$ is monotone non-decreasing.
}
\arxiv{\prooflemmatransferfunctionmonotonicity}

We use Lemma~\ref{lemma:transfer-function-monotonicity} to claim the following corollary, which essentially states that if a policy change universally lowers loss \emph{and} risk, then it is possible to adjust the price so that the profits go up and externalities go down.

\begin{corollary}\label{cor:transfer}
Let $(y,c), (y',c')$ be such that (a) $ye^{-c} \leq y'e^{-c'}$ and (b) for all $k$ in the support of $D_k$, $\ell(k,y,c) +c\geq \ell(k,y',c')+c'$ and $\risk(k,y,c) \geq \risk(k,y',c')$. Then for all $p$ and all $D_v$, there exists a $p'$ such that:
$$\rev_{D_v \times  D_k}(y',c',p') \geq \rev_{D_v \times D_k}(y,c,p),$$
$$\ext_{D_v \times D_k}(y',c',p') \leq \ext_{D_v\times D_k}(y,c,p).$$
\end{corollary}
\saveproof{cortransfer}{cor:transfer}{
First, consider setting $p':= p-c+c'$. Then the profits generated per sale are equal under $s = (y,c,p)$ and $s' = (y',c',p')$. Observe also that the probability of sale is \emph{at least as large} under $s'$ as $s$, as we have $u(t,s) = v - p +c - c- \ell(k,s)\leq v - p'+c' - c' -\ell(k,s') = u(t,s')$ for all $t$. Therefore we have that $\rev_{D_v \times D_k}(y',c',p-c+c') \geq \rev_{D_v \times D_k}(y,c,p)$. But unfortunately we can't yet say anything about the externalities. Indeed, the problem might be that there are many additional consumers with poor effectiveness who previously did not purchase the item under $s$ but who now purchase it under $s'$ (recall Example~\ref{ex:ext-non-monotonicity}). So the plan from here is to raise the price until the probability of sale is back to its original level (clearly the profits must still be larger, as now the probabilities of sale match, but the profit-per-sale of our new scheme is better). We'll use Lemma~\ref{lemma:transfer-function-monotonicity} to claim that the set of consumers who remain are only \emph{more} secure than what we started with.

So formally, raise the price $p'$ until the probability of sale for $s' = (y',c',p')$ is the same as $s$.\footnote{Note that we are assuming that for any desired probability $q$, we can set a price that sells with probability exactly $q$. When either $D_v$ or $D_k$ has no point masses, this is clearly true. When both have point masses, observe that if we set a price so that a positive mass of consumers are indifferent between purchasing the item and not, we will assume that we can have some buyers purchase the item and some not (as they are indifferent, either is a best response).} Now we have two schemes: $s = (y,c,p)$ and $s' = (y',c',p')$. Both sell the item with the same probability, $q$. If both schemes sold to \emph{exactly} the same $q$ fraction of consumers, then the lemma hypothesis that $\risk(k,s') \leq \risk(k,s)$ for all $k$ would suffice to let us claim that $\ext_{D_v \times D_k}(s') \leq \ext_{D_v \times D_k}(s)$. However, it could be a completely different $q$ fraction of consumers. Still, it turns out that because $y'e^{-c'} \geq ye^{-c}$, the fraction of consumers that purchase only have larger $k$. 

Indeed, observe that if some consumer $t=(v,k)$ purchases under $s$ but not $s'$, and some other consumer $t'=(v',k')$ purchases under $s'$ but not $s$, then we have:
$$g_{s,s'}(k)< 0 < g_{s,s'}(k')$$
But by Lemma~\ref{lemma:transfer-function-monotonicity}, we know that $g_{s,s'}(\cdot)$ is monotone increasing, so $k' > k$. In particular, this means that \emph{every} consumer in the mass which purchased under $s$ but not $s'$ has lower $k$ than \emph{any} consumer which purchased under $s'$ but not $s$. As $\risk(\cdot,s)$ is clearly monotone decreasing in $k$, and the fraction of buyers purchasing under $s$ and $s'$ is the same, we conclude that we must have $\ext_{D_v\times D_k}(s') \leq \ext_{D_v \times D_k}(s)$ as desired.
}
\arxiv{\proofcortransfer}

Now we are ready to \arxiv{prove}\www{formally state} the extension lemma for large $k$. 

\begin{lemma}[Extension of Fine Policy]\label{lemma:fineextension}
Let $D_k$ be supported on $[T,\infty)$, where $T$ is such that a fine policy is optimal for $D_v \times \{T\}$ subject to profits $R$. Then there is a fine policy $s'$ with $\rev_D(s') \geq R$ such that for all profits-maximizing $s$ with $\rev_D(s) \geq R$, $\ext_D(s') \leq \ext_D(s)$. 
\end{lemma}
\saveproof{lemmafineextension}{lemma:fineextension}{
First, observe that we necessarily have $T > 1$ if the hypothesis is to hold, by Theorem~\ref{thm:opt-homogeneous}. 

Consider any proposed optimal policy $s = (y,c,p)$. Let $s^*=(y^*,0,p^*)$ denote the optimal fine policy for $D_v \times \{T\}$ subject to profits $R$. Then maybe $\ell(T,s)+c \geq  \ell(T,s^*)$. If so, let $s' = (y',0,p-c)$ be such that $\ell(T,s')=\ell(T,s)+c$. As $T > 1$, observe that decreasing $c$ \emph{decreases} $\ell(T,s)+c$. Therefore, decreasing $c$ to $0$ results in $y' \geq y$ for the equality to hold. As such, $s$ and $s'$ satisfy the hypotheses of Lemma~\ref{lemma:transfer-function-monotonicity}, and we can conclude that $\ell(k,s') \leq \ell(k,s)+c$ for all $k \geq T$. We now just need to show that $\risk(T,s') \leq \risk(T,s)$. This is surprisingly tricky, and carried out in the subsequent paragraph.

Indeed, observe that $(y,c)$ is \emph{some} partial policy, and setting the price $\hat{p}$ which maximizes profits yields $R(y,c):=\rev_{D_v \times \{T\}}(y,c,\hat{p})$ on population $D_v \times \{T\}$. Observe that as $\ell(T,s)+c\geq \ell(T,s^*)$, we have $R(y,c) \leq R$ (otherwise $s^*$ would not be optimal for $D_v \times \{T\}$ subject to constraint $R$, as we could increase $y^*$). We can now ask what is the optimal policy for $D_v \times \{T\}$ subject to constraint $R(y,c)$? By Lemma~\ref{lem:hardestcase}, we claim it must be a fine policy, and that this fine policy is exactly $(y',0,\hat{p})$. 
To see this, we use Proposition~\ref{prop:thresholds}, which asserts that $T(D_v,R(y,c)) \leq T(D_v,R) = T$. As $T \geq T(D_v,R(y,c))$ now, we conclude that a fine policy must be optimal for $D_v \times \{T\}$ subject to profit constraint $R(y,c)$. This policy $s'$ must have $\ell(T,s')=\ell(T,s)+c$ (if it is bigger, then the profit will be $< R(y,c)$. If it is smaller, then the loss should be increased to get more profit). $(y',0,\hat{p})$ is exactly this policy. As it is optimal for $D_v\times \{T\}$. \emph{We may now conclude that $\risk(T,s') \leq \risk(T,s)$.} In particular, this necessarily implies that the risk is $\frac{1}{y'T}$ (because $c = 0$, the only other alternative would be to have risk $=1$, which is clearly not $\leq e^{-c}$ for $c > 0$). To conclude that $\risk(k,s') \leq \risk(k,s)$ for all $k \geq T$, simply observe that we must now have $\risk(k,s') = \frac{1}{y'k} \leq \frac{1}{y'T} \leq e^{-c}$, and also $\frac{1}{y'k} \leq \frac{1}{yk}$ as $y' \geq y$. So whether or not a consumer with effectiveness $k$ has $h^*(k,s) > 0$, the risk $\frac{1}{y'k}$ is better. Now we can apply Lemma~\ref{lemma:transfer-function-monotonicity}: we have come up with a new policy where everyone's risk is (weakly) lower, and everyone's loss is (weakly) lower, so we can increase the price until the probability of sale is the same, and this will (weakly) increase the profit and (weakly) decrease the risk.

So now we've covered the case that $\ell(T,s)+c\geq\ell(T,s^*)$. We just now need to consider the case that $\ell(T,s) +c< \ell(T,s^*)$. This is the only case where we'll assume that the $s$ we started with was profits-maximizing. Observe that if $\ell(T,s)+c < \ell(T,s^*)$, then there exists a price $p'$ such that $(y,c,p')$ generates profits strictly exceeding $R$. Indeed, if $\ell(T,s) +c< \ell(T,s^*)$, then $\ell(k,s) +c< \ell(T,s^*)$ for all $k\geq T$, so the distribution of $v - \ell(k,s) - c$ \emph{strongly} stochastically dominates the distribution of $v - \ell(T,s^*)$ in the following sense: for any probability $q$, the value $v_q$ such that $\Pr[v-\ell(k,s)-c\geq v_q]=q$ exceeds $w_q$ such that $\Pr[v-\ell(T,s^*) \geq w_q] = q$. As there exists a price $p$ such that $p \cdot \Pr[v - \ell(T,s^*)\geq p] \geq R$, that same probabilility of sale with a strictly increased price $p'$ guarantees that $\Pr[v-\ell(k,s)-c \geq p']\cdot p'> R$. 

Finally, Lemma~\ref{lemma:inv-transform} (stated below) implies that every optimal policy $s$ for $D_v \times D_k$ subject to profits constraint $R$ has $\rev_{D_v\times D_k}(s) = R$. This is because for all policy $s$ where the profit $> R$, there is a $\varepsilon > 0$ such that we can construct a policy $s'$ with profit $\geq R$ but strictly less externalities. This should perhaps not be surprising, as intuitively one should be able to decrease externalities at the cost of a little $\varepsilon$ (although one should be careful to do this properly this for arbitrary $\varepsilon$). Therefore, the optimal policy we started with achieved profits $R$, while the previous paragraph observes that there necessarily exists a price which achieves profits $>R$. So our original policy must not be profit-maximizing.}
\arxiv{\prooflemmafineextension}
\www{
\begin{proof}[Proof Sketch]
Consider any proposed optimal policy $s = (y, c, p)$. We first consider the case where $\ell(T, s) + c \geq \ell(T, s^*)$ where $s^*$ is the optimal fine policy on $D_v \times \{T\}$. If that is the case, then consider a fine policy $s' = (y', 0, p - c)$ where $\ell(T, s') = \ell(T, s) + c$. Observe we must have $y' \geq y$ since we can only obtain equality by increasing fines, then by \Cref{lemma:transfer-function-monotonicity}, we have $\ell(k, s') \leq \ell(k, s) + c$ for all $k$ in the support of $D_k$ and the profit under $s'$ can only be higher. 

Now, we need to show that the risk is only lower for all $k \geq T$. First, as $\ell(T, s') \geq \ell(T, s^*)$ we conclude that $y' \geq y^*$. Next, we can argue (see~\cite{venturyne2019externality} for the full proof) that we must have $\frac{1}{y^*T} \leq e^{-c}$. This allows us to conclude that $\frac{1}{y'T} \leq e^{-c}$, and we already have that $y' \geq y$ so $\frac{1}{y'k} \leq \frac{1}{yk}$. This gives us that $\risk(k,s') \leq \risk(k,s)$ for all $k \geq T$. Then Corollary~\ref{cor:transfer} let's us claim that there exists a price $p'$ for which both the profits and the externalities are better for $(y',0,p')$ than $s$. 

For the case where $\ell(T, s) + c < \ell(T, s^*)$, observe that $\ell(D_k, s) + c$ is strictly stochastic dominated by $\ell(D_k, s^*)$; therefore, there is a price $\hat p$ where the profit strictly higher than $R$. \Cref{lemma:inv-transform} bellow, allows us conclude that if $s$ is optimal, then $\rev_{D_v \times D_k}(s) = R$; otherwise, we can compromise $\varepsilon > 0$ fraction of the profit to strictly improve externalities. This implies that if the policy $s$ we start with is optimal, then it is not profits-maximizing.
\end{proof}
}

\begin{definition}[Invariant Transformation]
Given a policy $s = (y, c, p)$, define
$$\INV(s,\alpha) := \big(ye^{(p-c)(1-\alpha)}, \alpha c + (1-\alpha) p,p\big)$$
where $\alpha \in [0, \frac{p}{p-c}]$.
\end{definition}

\begin{lemma}[Invariant Property]\label{lemma:inv-transform}
Let $s' = \INV(s, \alpha)$, then for all $k \in \mathbb R^+$
\begin{itemize}
    \item $\hygiene(k, s') = \hygiene(k, s)$.
    \item $\ell( k, s') = \ell(k, s)$.
    \item $u(t,s) = u(t, s')$. 
\end{itemize}
In addition,
\[\rev_D(s') = \alpha \rev_D(s)\]
\[\ext_D(s') = e^{-(1-\alpha)(p - c)}\ext_D(s, p)\]
\end{lemma}
\saveproof{lemmainvtransform}{lemma:inv-transform}{
First note that for all types $t$, their optimal effort under policy $s'$ is the same as under policy $s$:
\begin{align*}
h^*(k, s') &= \bigg (\frac{\ln(yk) + (1-\alpha)(p-c) - \alpha c - (1-\alpha)p}{k}\bigg)^+\\
&=\bigg(\frac{\ln(yk) -c}{k}\bigg)^+= h^*(k, s)
\end{align*}

Also:
\begin{align*}
\risk(k, s') &= e^{- \alpha c - (1-\alpha)p - kh^*(k, s')}\\
& = e^{-(1-\alpha)(p-c)}\cdot e^{-kh^*(k,s')-c}=e^{- (1-\alpha)(p-c)}  \risk(k, s)
\end{align*}

And: 
\begin{align*}
\ell(k, s') &= ye^{(1-\alpha)(p-c)}\risk(k, s') + h^*(k, s') \\
&= y\risk(k,s)+h^*(k,s)= \ell(k, s)
\end{align*}

Which implies that:
\begin{align*}
    u(t,s') = v - \loss(t,s') - p = v - \loss(t,s) -p= u(t,s).
\end{align*}
We can conclude that a type $t$ purchases in policy $s$ iff type $t$ purchases in policy $s'$. Therefore:
\begin{align*}
\rev_D(s') &= (p - \alpha c - (1-\alpha) p) Pr_{t \leftarrow D}[u(t,s') \geq 0]\\
&= \alpha(p - c) Pr_{t \leftarrow D}[u(t,s)\geq 0]\\
&= \alpha \rev_D(s)
\end{align*}
and finally:
\begin{align*}
    \ext_D(s' ) &= E_{t \leftarrow D}[\risk(t, s') | u(t,s') \geq 0]\\
    &= e^{-(1-\alpha)(p - c)}E_{t \leftarrow D}[\risk(t, s) | u(t,s) \geq 0]\\
    &= e^{-(1-\alpha)(p - c)}\ext_D(s)
\end{align*}
which concludes the proof.
}
\arxiv{\prooflemmainvtransform}
\www{
\begin{proof}[Proof Sketch]
By applying the definition of buyer's efficiency, we can see that for all $k$, $h^*(k, s') = h^*(k, s)$ which implies $\ell(k, s') = \ell(k, s)$, and $\risk(k, s') = e^{-(1-\alpha)(p-c)}\risk(k, s)$.
\end{proof}
}
This concludes the proof of bullet two of Theorem~\ref{thm:opt-homogeneous}.

\subsection{Example: The Profits-Maximizing Qualification is Necessary}\label{sec:newexample}
In this section we provide the example promised in \Cref{prop:newexample}. \www{Refer to \Cref{sec:homogeneous-distribution} of \cite{venturyne2019externality} for omitted proofs.} Consider the following distribution, and profits constraint $R := 0.5$:
\begin{align*}
D_v = \begin{cases}
v_1 = 1 & \text{w. p. $\frac{1}{2}$}\\
v_2 = 16/15 & \text{w. p. $\frac{1}{2}$}
\end{cases}
\end{align*}
\begin{align*}
D_k = \begin{cases}
k_1 = 3 & \text{w. p. $\frac{1}{2}$}\\
k_2 = x \rightarrow \infty & \text{w. p. $\frac{1}{2}$}
\end{cases}
\end{align*}

Above, $x$ will be finite, but approaching $\infty$, and $\varepsilon$ will be finite but approaching $0$). The proposition will follow from the following sequence of claims. First, we will establish bullet one for $T:=3$.

\begin{claim}\label{claim:newexample1} A fine policy is optimal for $D_v \times \{3\}$. 
\end{claim}

\saveproof{claimnewexampleone}{claim:newexample1}{
To establish the claim, we will directly find the $c^*$ which is maximal among those such that $\rev(0,c,p)\geq 0.5$, and observe that $T=1+1/c^*$. 

Indeed, if $c = 0.5$, then $D_v - c$ takes value $\geq 1/2$ with probability $1$, so profits $0.5$ is indeed achievable. However, for any $c > 0.5$, $D_v-c$ takes value $w_1<1/2$ with probability $1/2$, and $w_2<1$ with probability $1/2$. So setting either price $w_1$ or $w_2$ yields profits $<1/2$. So $c^*=1/2$ for this example, and $3 = 1+1/c^*$ as desired.
}
\arxiv{\proofclaimnewexampleone}

Bullet two now immediately follows, as $D_k$ is indeed supported on $[3,\infty)$. We now just need to find the optimal fine policy for $D_v \times D_k$, and establish a better policy that is not simple. We now search for the optimal fine policy. Such a policy might sell only to $(16/15, x)$, but then the profits is at most $4/5$, which is too little. Such a policy might sell only to $(16/15,x)$ and $(1,x)$. But since $x$ is finite, such a policy certainly charges price $<1$ (unless $y = 0$, in which case the policy sells to all four types), and sells with probability $\leq 1/2$, so the profits are also too small. Such a policy might sell to all four types, which we analyze below. Or it might sell to all types except $(1,3)$, which we analyze after.

\begin{claim}\label{claim:newexample2} The optimal fine policy $s$ which sells to all four types has $\ext_{D_v \times D_k}(s) \geq \frac{1}{2\sqrt{e}}$.
\end{claim}

\saveproof{claimnewexampletwo}{claim:newexample2}{
Such a policy necessarily has $\loss(3,y) \leq 1/2$, which means that we must have $\frac{\ln(3y)+1}{3} \leq 1/2$, or $y \leq \sqrt{e}/3$. Such a policy has externalities at least $\frac{1}{2} \cdot \frac{3}{3\cdot \sqrt{e}} = \frac{1}{2\sqrt{e}}$. 
}
\arxiv{\proofclaimnewexampletwo}

\begin{claim}\label{claim:newexample3} The optimal fine policy $s$ which sells to all types except $(1,3)$ has $\ext_{D_v \times D_k}(s) \geq e^{-1/5}/3$.
\end{claim}

\saveproof{claimnewexamplethree}{claim:newexample3}{
Such a policy certainly has $16/15 - \loss(3,y) \geq 2/3$, as we are now selling with probability $3/4$, so we must charge a price at least $2/3$ in order to get profits $\geq 1/2$. Observe that $16/15 = 2/3+2/5$, so we must now have $\loss(3,y) \leq 2/5$. That yields $\ln(3) + \ln(y) +1 \leq 6/5$, or $y \leq e^{1/5}/3$. So the externalities are at least $\frac{1}{3} \cdot \frac{3}{e^{1/5}3} = e^{-1/5}/3$. 
}
\arxiv{\proofclaimnewexamplethree}

\begin{corollary}\label{cor:newexample4} The optimal fine policy $s$ has $\ext_{D_v \times D_k}(s) \geq e^{-1/5}/3$.
\end{corollary}

Here's now some intuition for how we're going to design a better non-simple policy: given that we wish to sell to all types except $(1,3)$, we can set $y$ very close to $0$ and have $\risk(x,s) \approx 0$, because $x$ is so large. The remaining question is then whether we wish to use $y$ or $c$ to make the risk of $(16/15, 3)$ as small as possible. Note that we must keep their loss under $2/5 < 1/2$ (as above). But for $k = 3$, a loss of $1/2$ is \emph{exactly} the cutoff when it becomes more efficient to use a fine policy instead of a cost policy. So if we use $c$ instead, we can get the risk lower for the same loss.

\begin{claim}\label{claim:newexample5}Let $\varepsilon$ be such that $\frac{\ln(x)+1}{x} \leq \varepsilon$. Then set $c = 1/3-\varepsilon$, and $y = (2/5-c)e^{c}$. Then $\ext_{D_v \times D_k}(y,c,2/3+c) = \frac{2}{3yx} + e^{-1/3+\varepsilon}/3$ and $\rev_{D_v \times D_k}(y,c,2/3+c) = 1/2$.
\end{claim}

\saveproof{claimnewexamplefive}{claim:newexample5}{
$\loss(3,(y,c)) \leq ye^{-c} = (2/5-c)e^{c} \cdot e^{-c} = 2/5-c$. So $(16/15,3)$ is willing to pay $c+2/3$. 

$\loss(x,(y,c)) = \frac{\ln(xy) - c+1}{x} = \frac{\ln(x(2/5-c))+\ln(e^c) - c + 1}{x} = \frac{\ln(x) + \ln(2/5-c)+1}{x}$. Because $\frac{\ln(x)+1+\ln(2/5-c)}{x} \leq \frac{\ln(x)+1}{x} \leq  \varepsilon$, the loss is $\leq \varepsilon = 1/3-c$, so $(1,x)$ is willing to pay $c+2/3$. 
Finally, we just need to compute the externalities and profits. The profits are exactly $1/2$, as it sells with probability $3/4$ and achieves profit $2/3$ when selling. The externalities are exactly $\frac{2}{3} \cdot \frac{1}{yx} + \frac{1}{3}\cdot e^{-c}$. 
}

\arxiv{\proofclaimnewexamplefive}

Now, we just need to compare $e^{-1/5}/3$ and $e^{-1/3+\varepsilon}/3 + \frac{2}{3yx}$. Observe that as $x \rightarrow \infty$, $\varepsilon \rightarrow 0$ and $e^{-1/3+\varepsilon}/3$ approaches $e^{-1/3}/3 $. So $\frac{2}{3yx} +e^{-1/3+\varepsilon}/3 \rightarrow 0 + e^{-1/3}/3 < e^{-1/5}/3$, and the externalities are indeed lower.

As a sanity check, we'll show that $((2/5-c)e^{1/3-\varepsilon},1/3-\varepsilon,2/3+c)$ is not profits-maximizing (technically, \Cref{thm:opt-homogeneous} doesn't imply this, since we didn't prove that the scheme is optimal. But as this scheme is better than all fine policies, certainly the optimal policy is not simple, and therefore not profits-maximizing by \Cref{thm:opt-homogeneous}. So the fourth bullet is already proven).

\begin{claim}\label{claim:newexample6} $((2/5-c)e^{1/3-\varepsilon},1/3-\varepsilon,2/3+c)$ is not profits-maximizing.
\end{claim}

\saveproof{claimnewexamplesix}{claim:newexample6}{
The four quantities of value minus loss are equal to: $\{16/15 - \varepsilon, 1-\varepsilon, 1-\varepsilon, 14/15-\varepsilon\}$. The seller generates profits $1/2$ by setting price $2/3+c$. If instead they set price $3/5+c$, the item would sell with probability $1$ and yield profits $3/5$.
}

\arxiv{\proofclaimnewexamplesix}
\section{General Distributions: An Approximation}\label{sec:general}
In this section, we consider general distributions. Clearly, one should not expect a simple policy to be optimal in general. Given that simple policies are optimal for homogeneous populations, one might reasonably expect that simple policies are approximately optimal for general distributions by simply ignoring half of the population and targeting the half that is responsible for most of the externalities. This idea works in one direction: if the ``low k'' region is responsible for most of the externalities in the optimum solution, then using a cost policy for the entire distribution is a good idea: the high k consumers may have significantly higher risk than previously, but this doesn't outweigh the original risk from the low k region.

This idea fails horribly, however, if the ``high k'' region is responsible for most of the externalities in the optimum solution. The problem is that while we can choose a policy to exclusively target this subpopulation, any low k (think: $k=0$) consumers who choose to purchase anyway may have enormous risk in comparison to before (i.e. it could now be $1$ when it was previously $e^{-c}$ for large $c$). We first show that this intuition can indeed manifest in a concrete example by presenting a lower bound in Section~\ref{sec:lower-bound}. This rules out a \emph{single-criterion} approximation that satisfies the profits constraint exactly, and approximates the externalities. In Section~\ref{sec:general}, we present a bicriterion approximation which approximately satisfies the profits constraint and also approximately minimizes externalities. This approximation is our most technical result. As such, we provide mainly proof sketches to overview the key steps.

\subsection{Lower Bound on Heterogeneous Distributions}\label{sec:lower-bound}

The key insight for our example is to make the profits constraint so binding that the only way to match it exactly is for the entire population to purchase the item. Part of the population will have $k=0$, and part will have $k \rightarrow \infty$. With both $c$ and $y$, it will be feasible to get the $k \rightarrow \infty$ consumers to have risk essentially $0$, while the $k=0$ consumers will have reasonably small risk. But with either $c =0$ or $y=0$, one of these will be lost, which causes significant risk increase.

\begin{example}\label{ex:lower-bound}
Let $D_v$ be a point mass at $v_0 = 2e^{x/2} \cdot (x+e^{-x})$. Let $D_k$ be a distribution with two point masses, one at $k=0$ with probability $e^{-x/2}$, one at $e^{xe^{x/2}}$ with probability $1-e^{-x/2}$. Let $R:= v_0 - e^{-x}-x$. 
\end{example}

\begin{lemma}\label{lem:expart1} The policy $(1,x,R+x)$ achieves profit $R$ in Example~\ref{ex:lower-bound}, and has externalities $\leq e^{-x/2} \cdot e^{-x} + 1\cdot e^{-xe^{x/2}}$.
\end{lemma}
\begin{proof}
The utility of $(v_0,0)$ is exactly $v_0 - e^{-x} - R - x = 0$, so they will choose to purchase. $(v_0,e^{xe^{x/2}})$ has only larger utility, so they will purchase as well. Therefore, the profit is indeed $R$.

The externalities are computed simply as the probability of having consumer $(v_0,0)$ times their risk ($e^{-x}$) plus (upper bound on the) probability of consumer $(v_0,e^{xe^{x/2}})$ times their risk $(e^{-xe^{x/2}})$. 
\end{proof}

\begin{lemma}\label{lem:expart2} Any cost policy that achieves profit $R$ has externalities at least $e^{-x+1}$
\end{lemma}
\begin{proof}
The maximum security we can set and still have profit $R$ is $x+e^{-x}$. If we set this, then the risk of all consumers (which is now independent of $k$) is $e^{-x+e^{-x}} \geq e^{-x+1}$. 
\end{proof}

\begin{lemma}\label{lem:expart3} Any fine policy that achieves profit $R$ has externalities at least $e^{-x/2}$. 
\end{lemma}
\begin{proof}
To achieve profit $R$, the policy must sell to the entire population. The consumer with $k=0$ will not put in any effort, and therefore their risk will be one, and the externalities will be at least $e^{-x/2}$. 
\end{proof}

\begin{corollary}\label{cor:lb} For all $x$, there exists a distribution $D_v \times D_k$ and profits constraint $R$ such that the optimal policy is not simple, and any simple policy that satisfies profits constraints $R$ has externalities at least a factor of $x$ larger than the optimum.
\end{corollary}

Corollary~\ref{cor:lb} is the main result of this section. Clearly the distribution witnessing Corollary~\ref{cor:lb} is highly contrived and unrealistic. And clearly, the way to get around this is to allow for a slight relaxation in the profits constraint so that we don't have to sell to the entire market (indeed, even allowing to relax the constraint by a $(1-e^{-x/2})$ fraction in this case would suffice). So the subsequent section shows that by relaxing the profits constraint, an approximation guarantee is possible.
\subsection{A Bicriterion approximation}\label{sec:approximation}

Given the lower bound in \Cref{sec:lower-bound}, we show that simple policies guarantee a bicriterion approximation. As is traditional with worst-case approximation guarantees, our constants are not particularly close to $1$, but are still relatively small. This is not meant to imply that the seller should be happy with (e.g.) a $1/8$-fraction of the original profits, but more qualitatively to conclude that simple policies can reap many of the benefits of optimal ones (see~\cite{hartline2013mechanism} for further discussion about the role of approximation in mechanism design). As referenced previously, the proof of Theorem~\ref{thm:approx} is quite technical, \arxiv{so we sketch the key steps and left the proof to \Cref{sec:proof-approx}.}\www{so we sketch the key steps. The complete proof can be found on \cite{venturyne2019externality}.}

\begin{theorem}\label{thm:approx}
For all distributions $D$, and all policies $s$, there exists a simple policy $s'$ such that
\begin{equation*}
\rev_D(s') \geq \rev_D(s)/8,
\end{equation*}
\begin{equation*}
\ext_D(s') \leq 40/3\cdot \ext_D(s).
\end{equation*}
\end{theorem}
\begin{proof}[Proof Sketch]
Given an arbitrary policy $s = (y, c, p)$, consider the conditional distribution of buyers that purchase under $s$. If with constant probability a buyer has efficiency $k \leq 1$, then we output the cost policy $s_1 := (0, c + \ell(\sigma, s), p + \ell(\sigma, s))$ where $\sigma$ is chosen such that a buyer continues to purchase with constant probability. We can show that $c + \ell(\sigma, s)$ is sufficiently large such that $\risk(D_k, s') \leq \risk(D_k, s)$ with constant probability.

For the case where with constant probability a buyer has efficiency $k > 1$, we define a blowup of the fines such that with constant probability a buyer continues to purchase but with the hope that inefficient buyers stop to purchase. The blowup can fail in two conditions: (1) $D_k$ is not heavy tail, (2) $D_v$ is heavy tail. For (1), we cannot derive a significant blowup if $D_k$ is concentrated close to 1. For (2), we cannot drive inefficient buyers out of the market if they have high value. Either condition allow us to construct cost policies that give good externality guarantees.
\end{proof}
\section{Summary}
We propose a stylized model to study regulation of single item sales with negative externalities, from which neither the buyer nor seller suffer. We first show that a simple policy is optimal in homogenous markets: That is, for all $D_v$, $R$, there exists a cutoff $T$ such that when the effectiveness of consumers ranges in $[0,T]$, the optimal policy regulates only the product (and does not impose fines). Similarly, if all consumers have effectiveness in $[T,\infty)$, a policy which regulates only payments (via fines, and does not impose default security features) outperforms all profits-maximizing policies. Importantly, $T$ is not necessarily the cutoff at which the consumers are more effective than the producer (which would be $T=1$), but actually depends on the value distribution $D_v$ and profit constraint $R$. 

We then show in general markets that while a simple policy may not be optimal, one is always approximately optimal. In particular, we show that while no simple scheme can guarantee any finite approximation while satisfying the profit constraint exactly, a bicriterion approximation exist, which approximately satisfies the profit constraint and also approximately minimizes externalities. Going forward, we must better understand the effectiveness of consumers to decide which regulation strategy is more appropriate to approximately minimizes externalities.

While stylized, our model captures the key salient features of this problem. We chose to study the single seller/single item setting in order to isolate these features without bringing in additional complexities (and the numerous examples throughout our paper demonstrate that even the single seller/single item setting is quite rich). Now that our results develop this understanding, a good direction for future work is to consider competing sellers or multiple items.

\bibliographystyle{ACM-Reference-Format}
\bibliography{masterbib}

\arxiv{\appendix

\section{Proof of Theorem~\lowercase{\ref{thm:approx}}}\label{sec:proof-approx}
We define the necessary tools for the case where a constant fraction of the population that purchase has efficiency $k \leq 1$ in \Cref{sec:approx-low-efficiency} and proof approximation guarantees in \Cref{sec:proof-approx-low-efficiency}.  For the case where a constant fraction of the population that purchase has efficiency $k > 1$, we define the necessary tools in \Cref{sec:approx-high-efficiency} and proof approximation guarantees in \Cref{sec:proof-approx-high-efficiency}. In \Cref{sec:proof-of-approx}, we combine the approximation guarantees to complete the proof of \Cref{thm:approx}.

\paragraph{Notation} Policy $s$ induces a threshold $k_0(s) := \inf\{k | h^*(s, t) > 0\}$ of types with zero effort. Let $k_h(s) := \max\{1, k_0(s)\}$, and define the events:
\[A(s) := \{\text{$t\leftarrow D$ purchase under $s$}\}\]
$$B(s) := A(s) \cap \{\text{$t \leftarrow D$ has efficiency $k > k_h(s)$}\}$$
We define the buyer's value after regulation:
$$\pvalue(t, s) := v - \ell(t, s)$$
then event $A(s)$ is equivalent to $\pvalue(t, s) \geq p$ or $u(t, s) \geq 0$.

\paragraph{Space Partition} Given $(D, R)$, and some arbitrary policy $s = (y, c, p)$. Let's look over the distribution of buyers that purchase under $s$. More formally, we will consider the following partition of the probability space:
\begin{align*}
    \varepsilon_1 &= \underset{t \leftarrow D}{Pr}[k \leq k_0(s) | u(t, s) \geq 0]\\
    \varepsilon_2 &= \underset{t \leftarrow D}{Pr}[k_0(s) < k \leq 1 | u(t, s) \geq 0]\\
    \varepsilon_3 &= \underset{t \leftarrow D}{Pr}[k > k_h(s) | u(t, s) \geq 0]
\end{align*}
The key idea behind our approximation mechanism, \Cref{alg:approx}, consists of defining three transformations of $s$. The cost policy $\COST^1$, \Cref{eq:cost-1} in \Cref{sec:approx-low-efficiency}, guarantees a constant approximation proportional to $\varepsilon_1$ and $\varepsilon_2$. \Cref{alg:skill}, \Cref{sec:approx-high-efficiency}, denoted by $\FINE$ outputs a simple policy that guarantees a constant approximation proportional to $\varepsilon_3$.

\begin{algorithm}
\caption{$\APPROX$}
\label{alg:approx}
\hspace*{\algorithmicindent} \textbf{Input: $s = (y, c, p)$, $D$} \\
\begin{algorithmic}[1]
\STATE Let $\beta = \frac{1}{2}$
\IF {$\varepsilon_1 \geq \frac{1}{8}$}
	\STATE Output $\COST^1(s, 1)$
\ENDIF
\IF {$\varepsilon_2 \geq \frac{1}{8}$}
	\STATE Output $\COST^1(s, \varepsilon_1 + \varepsilon_2)$
\ELSE
	\STATE Output $\FINE(s)$
\ENDIF
\end{algorithmic}
\end{algorithm}
\subsection{Approximation with Low Efficiency Buyers}\label{sec:approx-low-efficiency}
Next, we define the tools to proof \Cref{thm:approx} for the case where $\varepsilon_1 + \varepsilon_2 = O(1)$ and postpone the proof to \Cref{sec:proof-approx-low-efficiency}. The following policy clearly meets the profit guarantees when $\varepsilon = O(1)$.

\begin{equation}\label{eq:cost-1}
\COST^1(s, \varepsilon) := (0, c + \ell(\sigma, s), p + \ell(\sigma, s))
\end{equation}
where $\varepsilon \in [0, 1]$, and we choose $\sigma$ such that $Pr_{t \leftarrow D}[v \geq \ell(\sigma, s) + p] = \varepsilon Pr_{t \leftarrow D}[v \geq \ell(k, s) + p]$.

\paragraph{Zero efficiency case} For $\COST^1(s, 1)$, we get a good approximation to the externalities whenever a constant fraction $\varepsilon_1$ of the consumers who purchase have $h^*(t,s) = 0$. This is because the externalities are at least $\varepsilon_1 \cdot e^{-c}$ under $s$, and our new externalities are just $e^{-c}$. So if $\varepsilon_1$ is big enough, we get our desired approximation, \Cref{corollary:eps-1}.

\paragraph{Non-Zero efficiency case} If $\varepsilon_2$ is big, then $\varepsilon_2 \frac{1}{y}$ is a good lower bound for $\ext_D(s)$. This implies we must target a cost proportional to $\ln y$. For $\COST^1(s, \varepsilon_2)$, we can argue $\ell(\sigma, s)$ is at least $1 + \ln y$ which implies good externality bounds, \Cref{corollary:eps-2}.

\subsubsection{Proof of Approximation with Low Efficiency Buyers}\label{sec:proof-approx-low-efficiency}

\begin{lemma}\label{lemma:cost-1}
Let $\varepsilon \in [0, 1]$, then
\[\rev_D(\COST^1(s, \varepsilon)) = \varepsilon\rev_D(s, \varepsilon)\]
\[\ext_D(\COST^1(s, \varepsilon)) = \frac{e^{-c - \ell(\sigma, s)}}{E_{t\leftarrow D}[\risk(t, s)|u(t, s) \geq 0]}\ext_D(s)\]
\end{lemma}
\begin{proof}
By our choice of $\sigma$,
$$Pr_{t \leftarrow D}[v \geq p + \ell(\sigma, s)] = \varepsilon Pr_{t \leftarrow D}[v \geq p + \ell(k, s)]$$
This implies the following bounds in the profit,
\begin{align*}
\rev_D(\COST^1(s, \varepsilon)) &= (p + \ell(\sigma, s) - c - \ell(\sigma, s))Pr_{t \leftarrow D}[v \geq p + \ell(\sigma, s)]\\
&= (p - c) \varepsilon Pr_{t \leftarrow D}[u(t, s) \geq 0]\\
&= \varepsilon \rev_D(s)
\end{align*}
For the externality,
\begin{align*}
    \frac{\ext_D(\COST^1(s, \varepsilon))}{\ext_D(s)} = \frac{e^{-c - \ell(\sigma, s)}}{E_{t\leftarrow D}[\risk(t, s) | u(t, s) \geq 0]}
\end{align*}
which concludes the proof.
\end{proof}

\begin{corollary}\label{corollary:eps-1}
    \begin{align*}
    \rev_D(\COST^1(s, 1)) = \rev_D(s)
\end{align*}
\begin{align*}
    \rev_D(\COST^1(s, 1)) \leq \frac{1}{\varepsilon_1} \ext_D(s)
\end{align*}
\end{corollary}
\begin{proof}
The profit bound follows directly from \Cref{lemma:cost-1}. For the externality,
\begin{align*}
    \ext_D(s) &\geq \varepsilon_1 E_{t\leftarrow D}[\risk(t, s) | u(t, s) \geq 0, k \leq k_0(s)]\\
    &= \varepsilon_1 e^{-c}
\end{align*}
then,
\begin{align*}
    \frac{\ext_D(\COST^1(s, 1))}{\ext_D(s)} &= \frac{e^{-c - \ell(\sigma, s)}}{\ext_D(s)}\\
    &\leq \frac{e^{-c}}{\varepsilon_1 e^{-c}}
\end{align*}
which concludes the proof.
\end{proof}

\begin{corollary}\label{corollary:eps-2}
Assume $\varepsilon_2 > 0$, then
\begin{align*}
    \rev_D(\COST^1(s, \varepsilon_2)) \geq (\varepsilon_1 + \varepsilon_2) \rev_D(s)
\end{align*}
\begin{align*}
    \ext_D(\COST^1(s, \varepsilon_2)) \leq \frac{1}{\varepsilon_2 e} \ext_D(s)
\end{align*}
\end{corollary}
\begin{proof}
We will first claim $\ell(\sigma, s) \geq \ell(1, s)$. Assume for contradiction $\ell(\sigma, s) < \ell(1, s)$, then it must be
\begin{align*}
    (\varepsilon_1 + \varepsilon_2) Pr_{t \leftarrow D}[v \geq p + \ell(k, s)] &= Pr_{t \leftarrow D}[v \geq p + \ell(\sigma, s)]\\
    &>  Pr_{t \leftarrow D}[v \geq p + \ell(1, s)]\\
    &\geq Pr_{t \leftarrow D}[v \geq p + \ell(k, s) | k \leq 1]\\
    &= \frac{Pr_{t \leftarrow D}[k \leq 1 | u(t, s) \geq 0]Pr_{t \leftarrow D}[u(t, s) \geq 0]}{Pr_{t \leftarrow D}[k \leq 1]}\\
    &\geq Pr_{t \leftarrow D}[k \leq 1 | u(t, s) \geq 0]Pr_{t \leftarrow D}[u(t, s) \geq 0]\\
    &= (\varepsilon_1 + \varepsilon_2) Pr_{t \leftarrow D}[v \geq p + \ell(k, s)] \Rightarrow\!\Leftarrow
\end{align*}
where the first equality follows by the choice of $\sigma$ in \Cref{eq:cost-1} and the second equality follows by Bayes' theorem.

Next, we bound the externality with \Cref{lemma:cost-1}. We use the fact  $\sigma \leq 1$ and since $\varepsilon_2 > 0$, it must be $\hygiene(1, s) > 0$.
\begin{align*}
\frac{\ext_D(\COST^1(s, \varepsilon_2))}{\ext_D(s)} &= \frac{e^{-c - \ell(\sigma, s)}}{E_{t\leftarrow D}[\risk(t, s)|u(t, s) \geq 0]}\\
&\leq \frac{e^{-c -\ell(1, s)}}{\varepsilon_2 E_{t\leftarrow D}[1/yk|u(t, s) \geq 0, k_0(s) < k \leq 1]}\\
&= \frac{ye^{-c - 1 - \ln y + c}}{\varepsilon_2 E_{t\leftarrow D}[1/k|u(t, s) \geq 0, k_0(s) < k \leq 1]}\\
&\leq \frac{1}{e \varepsilon_2}
\end{align*}
The profit bound follow directly from \Cref{lemma:cost-1} which concludes the proof.
\end{proof}

\subsection{Approximation with High Efficiency Buyers}\label{sec:approx-high-efficiency}

In this section, we define the tools to proof \Cref{thm:approx} for the case where $\varepsilon_3 = O(1)$. We will construct $\FINE(s)$, \Cref{alg:skill}, which targets the population that purchases under $s$ and have high efficiency $k > k_h(s)$. To construct $\FINE(s)$, we will further define three additional transformations described in this section.

\begin{algorithm}
\caption{$\FINE(s)$}
\label{alg:skill}
\hspace*{\algorithmicindent} \textbf{Input: $\beta \in [0, 1], s = (y, c, p)$, $D$} \\

\begin{algorithmic}[1]
\STATE $\varepsilon = Pr_{t \leftarrow D}[k > k_h(s) | \pvalue(t, s) \geq p(s)]$
\IF {$c \leq 1$}
    \STATE Output $\INV(s, p, p/(p - c))$
\ENDIF
\IF {$\BLOWUP(s)$ is good} 
        \STATE Output $\BLOWUP \circ \INV(s, 1/2)$
        \label{alg-alpha-good}
\ELSE
        \IF {$\sigma \geq 2$}
                \STATE Output $\HEAVY(s)$ \label{alg-tail}
        \ELSE
        		\IF {$\exists x \in [e^{-c}, 1], \COST^3_x(s)$ is good}
            		\IF {$ye^{-c} < 2$}
            		\STATE Output $\COST^1(s, 1)$   \label{alg-cost}
                \ELSE
				    \STATE Output $\COST^3_x(s)$ \label{alg-x-good}
				\ENDIF
			\ELSE
				\STATE Output $\BLOWUP \circ \INV(s, 1/2)$ \label{alg-x-bad}
			\ENDIF
        \ENDIF
\ENDIF 
\end{algorithmic}
\end{algorithm}

As motivation consider \Cref{ex:lower-bound}. Ideally, we would like to make the inefficient buyer to stop to purchase. $\FINE(s)$ will first consider a blowup of the fines, $\BLOWUP(s)$ in \Cref{def:skillful-transformation}, dependent on the population that purchase under $s$. The blowup cannot be too high; otherwise, the utility of the efficient buyer decreases too much, hurting profit.

If the inefficient buyer is still willing to purchase after the blowup, it must be because the buyer value distribution has a heavy tail. In that case, we can derive the transformation $\HEAVY(s)$, \Cref{def:heavy-transformation}, that leverages the tail of the distribution to impose high security regulation directly on the product. 

\subsubsection{Preliminaries}\label{sec:skill-extra-transformation}

Let $G(x | \mathcal E) := Pr[k \leq x | \mathcal E]$ denote the cumulative distribution function of efficiency conditioned on event $\mathcal E$.  We will assume $\varepsilon = Pr_{t \leftarrow D}[k > k_h(s) | \pvalue(t, s) \geq p(s)]$. The parameter $\beta \in [0, 1]$ parameterize the fraction of the profit we are willing to compromise.

\subsubsection{Blowup}\label{sec:blowup}
\begin{definition}[Blowup transformation]\label{def:skillful-transformation}
We define,
\begin{equation}\label{eq:sigma}
    \sigma := \max \bigg \{1, G^{-1}(1-\beta | B(s))\bigg\}
\end{equation}
\begin{equation}
    \BLOWUP(s) := (qye^{c(\sigma - 1)}, 0, p-c)
\end{equation}
where
$$q := \inf\{x \geq 1 | Pr_{t \leftarrow D}[A(x y e^{\sigma - 1}, 0, p - c)] \leq Pr_{t \leftarrow D}[\pvalue(t, s) \geq p]\}$$
We define $y_{sk}(s)$ as the fine of $\BLOWUP(s)$,
\begin{align*}
    y_{sk}(s) := qye^{c(\sigma - 1)}
\end{align*}
\end{definition}

We define $\bar k(s)$ as the efficiency $k$ such that $\risk(\bar k(s), \BLOWUP(s)) = e^{-c}$,
$$\bar k(s) := \frac{e^{c}}{y_{sk}(s)}$$

By construction, if $\beta$ is big, then $\BLOWUP(s)$ will always provide good profit guarantees, \Cref{claim:blowup-profit}. For the externalities, we can also ensure $\BLOWUP(s)$ will provide good externality guarantees for the population that used to purchase under $s$ and had efficiency $k \geq \bar k(s)$, \Cref{claim:blowup-high-k}. This is because $y_{sk}(s) \geq y$; therefore, if $h^*(k, s) > 0$, we must have $\risk(k, \BLOWUP(s)) \leq \risk(k, s)$. If $h^*(k, s) = 0$, then we can have $\risk(k, \BLOWUP(s)) > \risk(k, s)$ if $k$ is small. $\bar k(s)$ precisely captures this phase change such that if $k \geq \bar k(s)$, then $\risk(k, \BLOWUP(s)) \leq \risk(k, s)$ and if $k < \bar k(s)$, $\risk(k, \BLOWUP(s)) > \risk(k, s)$.

To proof approximation bounds for $\BLOWUP(s)$ for the population where $k < \bar k(s)$, we first discuss the case where $\sigma \geq 2$. We discuss the case where $\sigma < 2$ on \Cref{sec:small-blowup}.

Assuming $\sigma$ is sufficiently large, $\BLOWUP(s)$ would still fail to provide good externality guarantees if the probability that a buyer with efficiency $k < \bar k(s)$ purchase under $\BLOWUP(s)$ is high. 

Unfortunately, for arbitrary distributions, we should not expect buyers with efficiency $k < \bar k(s)$ will have a small contribution to externalities under $\BLOWUP(s)$. However, we can define an upper bound on their externalities that would be sufficient to proof externality guarantees for $\BLOWUP(s)$, \Cref{claim:blowup-good}.

\begin{definition}[Good Blowup]\label{def:good-blowup}
$\BLOWUP(s)$ is good if buyers with efficiency at most $\bar k(s)$ give a small contribution to externalities:
\begin{equation}
\begin{split}
E_{t \leftarrow D}[&\risk(t, \BLOWUP(s)) \cdot \mathbb I(k \leq \bar k(s))| A(\BLOWUP(s))] \leq \useequation{good-blowup}
\end{split}
\end{equation}
\end{definition}

\paragraph{Composition}
Observe the dependence on $p - c$ on the externality bound which can be arbitrarily large. This can be easily be solved by applying a composition of $\BLOWUP$ with $\INV$. This is because \Cref{lemma:inv-transform} states the probability space of $s$ and $\INV(s, \cdot)$ is the same and by sacrificing a constant fraction of the profit we reduce externalities by a factor of $O\big(\frac{1}{p - c}\big)$.

\paragraph{Heavy Tail}
If $\BLOWUP(s)$ fails to push inefficient buyers out of the market (buyers with efficiency $k < \bar k(s)$ cause high externalities), it must be because the value distribution $D_v$ has a heavy tail. In that scenario, we define the policy $\HEAVY(s)$ that impose high regulation in the product and ensure constant approximation ratio, \Cref{claim:blowup-bad}.

\begin{equation}\label{def:heavy-transformation}
    \HEAVY(s) := \bigg(0, \frac{1}{2}p_{\HEAVY}, p_{\HEAVY}\bigg)
\end{equation}
where
\begin{equation}\label{eq:price-heavy}
    p_{\HEAVY} := \ell(\bar k(s), \BLOWUP(s)) = \bigg(\frac{(1+c)qye^{c\sigma}}{e^{2c}}\bigg)
\end{equation}

\subsubsection{When the Blowup is Small}\label{sec:small-blowup}
For the case where $\sigma < 2$, we might hope cost policies can still provide a good approximation since it suggests $D_k$ has a short tail. Next, we construct a family of cost policies $\mathcal H(s)$ such that if no policy $s \in \mathcal H(s)$ gives good profit guarantees, we can proof externality guarantees for $\BLOWUP(s)$, \Cref{claim:cost-3-bad-blowup}.

Under $\BLOWUP(s)$, in order for a buyer to have risk $x \in [e^{-c}, 1]$, she must draw efficiency $k = \frac{1}{y_{sk}(s) x}$ and observe $k \leq \bar k(s)$.

We define an alternative expression for the loss of $k$ under fine $y$ and cost $c = 0$ as function of its risk $x$:
\begin{equation}\label{eq:loss-risk}
\ell^\risk(x, y) := (\ln 1/x + 1)xy
\end{equation}
Let $H(x)$ be the probability a value is greater or equal to the loss $\ell^\risk(x, y_{sk}(s))$.
\begin{equation}\label{eq-sale-inv}
    H(x) := Pr_{t \leftarrow D}\bigg[v \geq \bigg(\ln \frac{1}{x} + 1\bigg)x y_{sk}(s)\bigg]
\end{equation}

We define the cost policy $\COST^3_x(s)$ where the probability of sale under $\COST^3_x(s)$ is equivalent to $H(x)$,
$$\COST^3_x(s) := \bigg(0, \ln y, \bigg(\ln\frac{1}{x} + 1\bigg)xy_{sk}(s)\bigg)$$
We define the family of cost policies $\mathcal H(s)$,
$$\mathcal H(s) := \{\COST^3_x(s) : x \in [e^{-c}, 1]\}$$
We can verify $Pr_{t \leftarrow D}[A(\COST^3_x(s))] = H(x)$,
\begin{align*}
    Pr_{t \leftarrow D}\bigg[&\pvalue(t, \COST^3_x(s)) \geq \bigg(\ln\frac{1}{x} + 1\bigg)xy_{sk}(s)\bigg] =  H(x)
\end{align*}
because $\ell(k, \COST^3_x(s)) = 0$ for all $k$.

Since $D_k$ has a short tail ($\sigma$ is small), all policies in $s \in \mathcal H(s)$ would give good externality guarantees when compared to $s$, \Cref{claim:cost-3-good}; however, the probability of sale might be too small to get good profit guarantees. Bellow, we define a sufficient condition to proof a policy in $\mathcal H$ would give good profit guarantees and the formal statement is given in \Cref{claim:cost-3-good}.
\begin{definition}[Good $\COST^3_x(s)$]\label{def:good-sale-invariant}
We define $\COST^3_x(s)$ as good if its probability of sale is bigger than
$$H(x) \geq \useequation{good-H(x)}$$
\end{definition}

If \Cref{def:good-sale-invariant} is not satisfied for any policy in $\mathcal H$, it implies the probability of sale to buyers with efficiency $k < \bar k(s)$ under $\BLOWUP(s)$ is bounded. This implies, we can directly bound the externalities contributed by $k < k(s)$ and proof externality guarantees for $\BLOWUP(s)$, \Cref{claim:cost-3-bad}.

\subsubsection{Proof of Approximation for High Efficiency Buyers}\label{sec:proof-approx-high-efficiency}
In this section, we bound the profit and the externality for the output of \Cref{alg:skill}.

\begin{claim}\label{claim:blowup-profit}
\[Pr_{t \leftarrow D}[A(\BLOWUP(s))] \geq \beta \varepsilon Pr_{t \leftarrow D}[u(t, s) \geq 0]\]
\[\rev_D(\BLOWUP(s)) \geq  \beta \varepsilon \rev_D(s)\]
\end{claim}
\begin{proof}
We will first bound the probability that a type $t$ that purchase under policy $s$ and has efficiency at least $k_h(s)$, also purchase under policy $s' = \BLOWUP(s)$. Given that $t$ purchase under policy $s$, we must have $p \leq \pvalue(t, s)$, then
\begin{align*}
Pr_{t\leftarrow D}&[\pvalue(t, s') \geq p - c| B(s)] \geq Pr_{t \leftarrow D}[\pvalue(t, s')\geq \pvalue(t, s) - c | B(s)]\\
&=Pr_{t \leftarrow D}[v - \frac{1 + \ln y + \ln k - c}{k} - \frac{c\sigma }{k}\geq v - \frac{1 + \ln y + \ln k - c}{k} - c | B(s)]\\
&= Pr_{t \leftarrow D}[k \geq \sigma | B(s)] = \beta
\end{align*}
In the last step, we use the fact $\sigma = G^{-1}(1 - \beta | B(s))$.

Next, we lower bound the probability of sale and the profit under policy $s'$.
\begin{align*}
Pr_{t \leftarrow D}[\pvalue(t, s') \geq p - c] &\geq Pr_{t \leftarrow D}[u(t, s) \geq 0]\\
&\cdot Pr_{t \leftarrow D}[k > k_h(s)|u(t, s) \geq 0]\\
&\cdot Pr_{t \leftarrow D}[\pvalue(t, s') \geq p - c | B(s)]\\
&\geq \beta\varepsilon Pr_{t \leftarrow D}[u(t, s) \geq 0]
\end{align*}
\begin{align*}
\rev_D(s') &=  (p - c) Pr_{t \leftarrow D}[\pvalue(t, s') \geq p - c]\\
&\geq \beta \varepsilon (p - c)Pr_{t \leftarrow D}[u(t, s) \geq 0]\\
&=  \beta \varepsilon \rev_D(s)
\end{align*}
which concludes the proof.
\end{proof}

\begin{claim}\label{claim:blowup-high-k}
\begin{align*}
E_{t \leftarrow D}[\risk(t, \BLOWUP(s)) \cdot \mathbb I(k > \bar k(s))| A(\BLOWUP(s))]\leq \frac{1}{\varepsilon \beta}E_{t \leftarrow D}[\risk(t, s) \cdot \mathbb I(k > \bar k(s)) | u(t, s) \geq 0]
\end{align*}
\end{claim}
\begin{proof}
Let $s' = \BLOWUP(s)$. If $k > \bar k(s)$, we can have either $h^*(k, s) > 0$ or $h^*(k, s) = 0$. In the first case, by definition of $\BLOWUP(s)$, $y' \geq y$, and together with $h^*(k, s) > 0$ implies $\risk(k, s') \leq \risk(k, s)$. If instead, $h^*(k, s) = 0$, $r(k, s) = e^{-c}$, but since $k > \bar k(s)$, $r(k, \BLOWUP(s)) \leq e^{-c}$. So in all cases, $r(k, s') \leq r(k, s)$.

Still by construction of $\BLOWUP(s)$, the probability $\pvalue(t, \BLOWUP(s)) \geq p - c$ is at most the probability $u(t, s) \geq 0$. By a similar argument in \Cref{cor:transfer}, $g_{s, s'}(k)$ is monotone decreasing which implies the set of types that start to purchase under $s'$ can only be more efficient than the set of types that stops to purchase. We can conclude
\begin{align*}
    E_{t \leftarrow D}[\risk(t, s') \cdot \mathbb I(k > \bar k(s), u(t, s') \geq 0)] \leq E_{t \leftarrow D}[\risk(t, s) \cdot \mathbb I(k > \bar k(s), u(t, s) \geq 0)]
\end{align*}
We can re-write the expectations as
\begin{align*}
E_{t \leftarrow D}[&\risk(t, s') \cdot \mathbb I(k > \bar k(s), u(t, s') \geq 0)] = E_{t \leftarrow D}[\risk(t, s')\cdot \mathbb I(k > \bar k(s)) | u(t, s') \geq 0] Pr_{t \leftarrow D}[u(t, s') \geq 0]
\end{align*}
By \Cref{claim:blowup-profit}, $Pr_{t \leftarrow D}[u(t, s') \geq 0] \geq \beta \varepsilon Pr_{t \leftarrow D}[u(t, s) \geq 0]$, then
\begin{align*}
    E_{t \leftarrow D}[\risk(t, s')\cdot \mathbb I(k > \bar k(s)) | u(t, s') \geq 0] \leq E_{t \leftarrow D}[\risk(t, s) \cdot \mathbb I(k > \bar k(s)) | u(t, s) \geq 0] \frac{Pr_{t \leftarrow D}[u(t, s) \geq 0]}{\beta \varepsilon Pr_{t \leftarrow D}[u(t, s) \geq 0]}
\end{align*}
which concludes the proof.
\end{proof}

\begin{claim}\label{claim:blowup-good}
If $\BLOWUP(s)$ is good, then
\begin{equation}
\ext_D(\BLOWUP(s)) \leq \bigg ( \frac{4 (p - c)}{(1-\beta)(1+c)\varepsilon} + \frac{1}{\beta\varepsilon} \bigg)\ext_D(s)
\end{equation}
\end{claim}
\begin{proof}
By combining Definition \ref{def:good-blowup}, Claim \ref{claim:blowup-high-k} that bounds the expected risk of buyers with efficiency in the intervals $[0, \bar k(s)]$, $(\bar k(s), \infty)$ respectively, the result follows.
\end{proof}

\begin{claim}\label{claim:blowup-bad}
If $\BLOWUP(s)$ is bad and $\sigma \geq 2$, $c \geq 1$, then 
$$\rev_D(\HEAVY(s)) \geq  \beta \varepsilon \rev_D(s)$$
$$\ext_D(\HEAVY(s)) \leq \frac{4}{(1-\beta)(c + 1)\varepsilon}\ext_D(s)$$
\end{claim}
\begin{proof}
Let $n = \frac{4 (p - c)}{(1-\beta)(1+c)\varepsilon}$, $s' = \BLOWUP(s)$. We first bound the probability of sale of $\HEAVY(s)$.

By \Cref{eq:price-heavy}, $p_{\HEAVY} = \ell(\bar k(s), s')$ which implies
\begin{align*}
Pr_{t \leftarrow D}[v \geq \ell(\bar k(s), s') + p(s')] &\leq Pr_{t \leftarrow D}[v \geq \ell(\bar k(s), s')]\\
&\leq Pr_{t \leftarrow D}[u(t, \HEAVY(s)) \geq 0]
\end{align*}
where the last inequality follows from $\pvalue(t, \HEAVY(s)) = v$ since $\HEAVY(s)$ is a cost policy.

For a fixed price $p$, the probability of sale must decrease as we decrease efficiency, then
\begin{align*}
Pr_{t \leftarrow D}[A(\HEAVY(s))] &\geq Pr_{t \leftarrow D}[v \geq \ell(\bar k(s), s') + p(s')]\\
&\geq Pr_{t \leftarrow D}[u(t, s') \geq 0| k \leq \bar k(s)]&\\
&= \frac{Pr_{t \leftarrow D}[k \leq \bar k(s) | u(t, s') \geq 0]Pr_{t \leftarrow D}[u(t, s') \geq 0]}{Pr_{t \leftarrow D}[k \leq \bar k(s)]}
\end{align*}
where the equality comes from Bayes' theorem.

Next, we bound $Pr_{t \leftarrow D}[k \leq \bar k(s) | u(t, s') \geq 0]$. The inequality comes from the fact $\risk(t, s') \leq 1$,
\begin{align*}
E_{t \leftarrow D}[\risk(t, s') \cdot \mathbb I(k \leq \bar k(s)) | u(t, s') \geq 0] &\leq E_{t \leftarrow D}[\mathbb I(k \leq \bar k(s)) | u(t, s') \geq 0]\\
&= Pr_{t \leftarrow D}[k \leq \bar k(s) | u(t, s') \geq 0]
\end{align*}
This implies,
\begin{align*}
    Pr_{t \leftarrow D}[A(\HEAVY(s))] &\geq Pr_{t \leftarrow D}[k \leq \bar k(s) | u(t, s') \geq 0]Pr_{t \leftarrow D}[u(t, s') \geq 0]\\
    &\geq E_{t \leftarrow D}[\risk(t, s') \cdot \mathbb I(k \leq \bar k(s)) | u(t, s') \geq 0]Pr_{t \leftarrow D}[u(t, s') \geq 0]
\end{align*}
Using the assumption that $s'$ is bad, and \Cref{claim:blowup-profit} to lower bound $Pr_{t \leftarrow D}[u(t, s') \geq 0]$, we get \Cref{eq:lower-bound-heavy-sale}
\begin{equation}\label{eq:lower-bound-heavy-sale}
\begin{split}
    Pr_{t \leftarrow D}[u(t, \HEAVY(s)) \geq 0] \geq \beta\varepsilon n Pr_{t \leftarrow D}[u(t, s) \geq 0] \ext_D(s)
\end{split}
\end{equation}
Next, we lower bound the profit of $\HEAVY(s)$. In the first equality, we use the definition of $p_\HEAVY$, \Cref{eq:price-heavy}. In the first inequality, we lower bound the probabily of sale of $\HEAVY(s)$, \Cref{eq:lower-bound-heavy-sale}. In the second inequality, we use the fact $\rev_D(s) = (p - c) Pr_{t \leftarrow D}[u(t, s) \geq 0]$.
\begin{align*}
\rev(\HEAVY(s)) &=  \frac{p_{\HEAVY}}{2} Pr_{t \leftarrow D}[v \geq p_{\HEAVY}]\\
&= \frac{(1+c)qye^{c\sigma}}{2e^{2c}}Pr_{t \leftarrow D}[v \geq p_{\HEAVY}]\\
&\geq \frac{(1+c)ye^{c\sigma}}{2e^{2c}}\beta\varepsilon n Pr_{t \leftarrow D}[u(t, s) \geq 0] \ext_D(s)\\
&\geq \beta\varepsilon n\bigg(\frac{(1+c)ye^{c\sigma}}{2(p - c)e^{2c}} \bigg)\ext_D(s)\rev_D(s)
\end{align*}
Next, we bound $E_{t \leftarrow D}[1/k|B(s)]$ and $\ext_D(s)$,
\begin{proposition}\label{prop:heavy}
\begin{align*}
E_{t \leftarrow D}[1/k | u(t, s) \geq 0, k > k_h(s)] \geq \frac{1-\beta}{\sigma}
\end{align*}
\begin{align*}
\ext_D(s) \geq \frac{\varepsilon (1-\beta)}{y\sigma}
\end{align*}
\end{proposition}
\begin{proof}
\begin{align*}
E_{t \leftarrow D}[1/k | B(s)] &\geq Pr_{t \leftarrow D}[k \leq \sigma | B(s)] E_{t \leftarrow D}[1 / k | u(t, s) \geq 0, k \in (k_h(s), \sigma]]\\
&\geq \frac{(1-\beta)}{\sigma}
\end{align*}
By definition of $\ext_D(s)$, and the previous bound
\begin{align*}
    \ext_D(s) &\geq \varepsilon E_{t \leftarrow D}[1/yk | B(s)]\\
    &\geq \frac{\varepsilon(1-\beta)}{y\sigma}
\end{align*}
\end{proof}
By \Cref{prop:heavy},
\begin{align*}
\rev(\HEAVY(s))&\geq \frac{\beta\varepsilon^2 n }{2}\bigg(\frac{(1-\beta)(1+c)e^{c\sigma}}{\sigma(p-c)e^{2c}}\bigg)\rev_D(s)\\
&= \frac{\beta\varepsilon^2 n}{2}\bigg(\frac{(1-\beta)(1+c)e^{2c + (c-1)(\sigma-2) +(\sigma -2 )}}{\sigma(p - c)e^{2c}}\bigg)\rev_D(s)
\end{align*}
Using the fact $c \geq 1$, $\sigma \geq 2$, $e^\sigma/\sigma \geq e^2/2$ and substituting for $n$,
\begin{align*}
\rev(\HEAVY(s)) &\geq \frac{(1-\beta)\beta \varepsilon^2 n (1+c)e^{\sigma}}{\sigma 2 e^2 (p-c)}\rev_D(s)\\
&\geq \frac{\beta \varepsilon^2 n (1+c)}{2 e^2 (p-c) }\frac{(1-\beta)e^\sigma}{\sigma}\rev_D(s)\\
&\geq \beta\varepsilon \rev_D(s)
\end{align*}

Similarly, we bound the externalities,
\begin{align*}
\ext_D(\HEAVY(s)) &= e^{-p_{\HEAVY}/2}\\
&\leq e^{-\frac{(1+c)ye^{c\sigma}}{2e^{2c}}}\\
&= e^{-\frac{(1+c)ye^{2c + c(\sigma-2)}}{2e^{2c}}}\\
&\leq e^{-\frac{(1+c)ye^{c(\sigma-2)}}{2}}\\
&< \frac{2}{(1+c)ye^{c(\sigma-2)}}\\
&\leq \frac{2}{(1+c)ye^{(c-1)(\sigma - 2) + (\sigma - 2) }}
\end{align*}
Using the fact $c \geq 1$, $\sigma \geq 2$, $e^\sigma/\sigma \geq e^2/2$ and dividing and multiplying by $E_{t \leftarrow D}[1/k|B(s)]$,
\begin{align*}
\ext_D(\HEAVY(s)) &\leq \frac{2 e^2 \varepsilon E_{t \leftarrow D}[1/k|B(s)]}{(c + 1)y E_{t \leftarrow D}[e^{\sigma}/k|B(s)]}\\
&\leq \frac{4e^2}{(1-\beta)(c + 1)\varepsilon e^2}\ext_D(s)
\end{align*}
\end{proof}
where the second inequality follows from \Cref{prop:heavy}.

\begin{claim}\label{claim:fine-cost-1}
If $ye^{-c} < 2$, $\sigma < 2$ then
\[\rev_D(\COST^1(s, 1-\delta)) = (1-\delta)\rev_D(s)\]
\[\ext_D(\COST^1(s, 1 - \delta)) \leq \frac{4}{\varepsilon(1-\beta)} \ext_D(s)\]
\end{claim}
\begin{proof}
We will proof $ye^{-c} < 2$ implies $E_{t \leftarrow D}[1/yk \cdot \mathbb I(k > k_h)| u(t, s) \geq 0] \geq \frac{e^{-c}\varepsilon(1-\beta)}{4}$, and the statement follows directly by \Cref{lemma:cost-1}.

\begin{align*}
    E_{t \leftarrow D}[1/yk \cdot \mathbb I(k > k_h)| u(t, s) \geq 0] &> \frac{\varepsilon e^{-c}}{2}E_{t \leftarrow D}[1/k| u(t, s) \geq 0, k > k_h]\\
    &\geq \frac{\varepsilon e^{-c} (1 - \beta)}{2 \sigma}\\
    &> \frac{\varepsilon e^{-c} (1 - \beta)}{4}
\end{align*}
where the second inequality follows by \Cref{prop:heavy}.
\end{proof}

\begin{claim}\label{claim:cost-3-good}
Assume $ye^{-c} \geq 2$, and $\sigma \leq 2$. If for some $x \in [e^{-c}, 1]$, $\COST^3_x(s)$ is good, then
\[\rev_D(\COST^3_x(s)) \geq  \beta \varepsilon \rev_D(s)\]
\[\ext_D(\COST^3_x(s)) \leq \frac{2}{\varepsilon (1-\beta)} \ext_D(s)\]
\end{claim}
\begin{proof}
To bound $\bigg(\ln \frac{1}{x} + 1\bigg)x$ when $x \in [e^{-c}, 1]$, note it achieves its minimum value when $x = e^{-c}$ which implies $\bigg (\ln \frac 1 x + 1\bigg )x \geq (c+1)e^{-c}$. We can then bound the net profit of policy $\COST^3_x(s)$. More precisely, we claim $p_{\COST^3_x(s)} - \ln y \geq \frac{\bigg( \ln \frac{1}{x} + 1 \bigg)x y_{sk}(s)}{2}$. By definition $p_{\COST^3_x(s)} = \bigg (\ln \frac 1 x + 1\bigg )xy_{sk}(s)$ and we will show $p_{\COST^3_x(s)} \geq 2 \ln y$.
\begin{align*}
    \bigg (\ln \frac 1 x + 1\bigg )xy_{sk}(s) &\geq (c + 1)e^{-c} ye^{-c}e^{c\sigma}\\
    &\geq (c + 1)ye^{-c}\\
    &\geq 2c + ye^{-c}&\text{by the fact $ye^{-c} \geq 2$}
\end{align*}
Write $\ln y = \ln ye^{-c} + c$ and observe $ye^{-c} \geq 2 \ln ye^{-c}$ by the fact $ye^{-c} \geq 2$. We can conclude $2c + ye^{-c} \geq 2 (c + \ln ye^{-c}) = 2 \ln y$ which implies $p_{\COST^3_x(s)} - \ln y \geq \frac{\bigg(\ln \frac 1 x + 1 \bigg)xy_{sk}(s)}{2}$. 

By the fact $\COST^3_x(s)$ is good, $H(x) \geq \useequation{good-H(x)}$, then
\begin{align*}
\rev_D(\COST^3_x(s)) &= (p_{\COST^3_x(s)} - \ln y) H(x)\\
&\geq \frac{\bigg(\ln \frac 1 x + 1 \bigg)xy_{sk}(s)}{2} \frac{2\beta\varepsilon \rev_D(s)}{\bigg(\ln \frac 1 x + 1\bigg )xy_{sk}(s)}\\
&\geq \beta \varepsilon \rev_D(s)
\end{align*}
Next, we bound the externalities,
\begin{align*}
\frac{\ext_D(\COST^3_x(s))}{\ext_D(s)} &= \frac{1}{y \ext_D(s)}\\
&\leq \frac{y \sigma}{y \varepsilon (1-\beta)}\\
&< \frac{2}{\varepsilon (1-\beta)}
\end{align*}
where the first inequality follows from \Cref{prop:heavy} and the second inequality follows from $\sigma < 2$.
\end{proof}

\begin{claim}\label{claim:cost-3-bad}
If $\sigma \leq 2$, and $\forall x \in [e^{-c}, 1]$, $\COST^3_x(s)$ is bad, then
\begin{align*}
    E_{t \leftarrow D}[\risk(t, \BLOWUP(s)) \cdot \mathbb I(k \leq \bar k(s))|A(\BLOWUP(s))] \leq \frac{4 (p - c) }{\varepsilon(1-\beta)}\ext_D(s)
\end{align*}
\end{claim}
\begin{proof}
Define $\mu_{\risk}^s(x) = Pr_{k \leftarrow D_k}[\text{\risk(k, s) = x}]$ and observe we can sample $k$ by sampling a risk $x$ from the distribution $\mu_{\risk}^s$ and computing $k$. Let's compute $E_{t \leftarrow D}[\risk(t, s') \cdot \indicator{k \leq \bar k(s), u(t, s') \geq 0}]$ where below, in the first equality, we apply the tower rule by sampling a risk $x$ from $\mu_{\risk}^{\BLOWUP(s)}$. In the second equality, if type $k$ has risk $x < e^{-c}$, by definition of $\bar k(s)$, we must have $k > \bar k(s)$ which implies $\indicator{k \leq \bar k(s)} = 0$. In the first inequality, we use the fact $H(x)$ upper bounds the probability of sale for a buyer with risk $x$.
\begin{align*}
    E_{t \leftarrow D}[&\risk(t, \BLOWUP(s)) \cdot \indicator{k \leq \bar k(s), A(\BLOWUP(s))}]\\
    &= E_{x \leftarrow \mu_{\risk}^{\BLOWUP(s)}}[E_{t \leftarrow D}[\risk(t, \BLOWUP(s)) \cdot \indicator{k \leq \bar k(s), A(\BLOWUP(s))}|\risk(k, \BLOWUP(s)) = x]] \\
    &= E_{x \leftarrow \mu_{\risk}^{\BLOWUP(s)}}[x \cdot  \indicator{x \in [e^{-c}, 1]} \cdot E_{t \leftarrow D}[\indicator{A(\BLOWUP(s))}|\risk(k, \BLOWUP(s)) = x]]\\
    &= E_{x \leftarrow \mu_{\risk}^{\BLOWUP(s)}}[x \cdot \indicator{x \in [e^{-c}, 1]} \cdot Pr_{t \leftarrow D}[v \geq x y_{sk}(s)(1 + \ln 1/x) + p(\BLOWUP(s))]]\\
    &\leq E_{x \leftarrow \mu_{\risk}^{\BLOWUP(s)}}[x \cdot \indicator{x \in [e^{-c}, 1]} \cdot H(x)]
\end{align*}
For all $x \in [e^{-c}, 1]$, because $COST^3_x(s)$ is bad, we must have
\begin{align*}
H(x) &< \useequation{good-H(x)}\\
&\leq \frac{2 \beta \varepsilon (p - c) Pr_{t \leftarrow D}[u(t, s) \geq 0]}{xy_{sk}(s)}
\end{align*}
This implies,
\begin{align*}
    E_{t \leftarrow D}[&\risk(t, \BLOWUP(s)) \cdot \indicator{k \leq \bar k(s), A(\BLOWUP(s))}]\\
    &< E_{x \leftarrow \mu_{\risk}^{\BLOWUP(s)}}[x \cdot \indicator{x \in [e^{-c}, 1]} \cdot \frac{2\beta \varepsilon(p-c)Pr_{t \leftarrow D}[u(t, s) \geq 0]}{x y_{sk}(s)}]\\
    &= \frac{2\beta \varepsilon (p - c) Pr_{t \leftarrow D}[u(t, s) \geq 0]}{y_{sk}(s)}
\end{align*}
Bellow, we use the fact $Pr_{t \leftarrow D}[A(\BLOWUP(s))] \geq \beta \varepsilon Pr_{t \leftarrow D}[u(t, s) \geq 0]$, \Cref{claim:blowup-profit}.
\begin{align*}
    E_{t \leftarrow D}[&\risk(t, \BLOWUP(s)) \cdot \indicator{k \leq \bar k(s), A(\BLOWUP(s))}]\\
    &= E_{t \leftarrow D}[\risk(t, \BLOWUP(s)) \cdot \indicator{k \leq \bar k(s)}|A(\BLOWUP(s))] Pr_{t \leftarrow D}[A(\BLOWUP(s))]\\
    &\geq \beta \varepsilon Pr_{t \leftarrow D}[u(t, s) \geq 0]E_{t \leftarrow D}[\risk(t, \BLOWUP(s)) \cdot \indicator{k \leq \bar k(s)}|A(\BLOWUP(s))]
\end{align*}
Combining the previous bounds, we have
\begin{align*}
    \frac{E_{t \leftarrow D}[\risk(t, \BLOWUP(s)) \cdot \indicator{k \leq \bar k(s)}|A(\BLOWUP(s))]}{\ext_D(s)} &\leq
    \frac{2 (p - c)}{y_{sk}(s) \ext_D(s)}\\
    &\leq \frac{2 y \sigma (p - c)}{y_{sk}(s)(1-\beta)\varepsilon} \leq \frac{4(p-c)}{(1-\beta)\varepsilon}
\end{align*}
\end{proof}
where the last inequality follows from \Cref{prop:heavy} which completes the proof.

\begin{claim}\label{claim:cost-3-bad-blowup}
If $\sigma \leq 2$, and $\forall x \in [e^{-c}, 1]$, $\COST^3_x(s)$ is bad then
\begin{equation}
\ext_D(\BLOWUP(s)) \leq \bigg ( \frac{4 (p - c) }{\varepsilon(1-\beta)} + \frac{1}{\beta\varepsilon} \bigg)\ext_D(s)
\end{equation}
\end{claim}
\begin{proof}
By combining \Cref{claim:cost-3-bad} and \ref{claim:blowup-high-k}, that bound the expected risk of buyers with efficiency in the intervals $[0, \bar k(s)]$, $(\bar k(s), \infty)$ respectively, the result follows.
\end{proof}

\subsection{Proof of \lowercase{\Cref{thm:approx}}}
\label{sec:proof-of-approx}

\begin{proof}[Proof of \Cref{thm:approx}]
If $\APPROX(s)$ outputs $\COST^1(s, 1)$, then by \Cref{lemma:cost-1}, the profit is
$$\rev_D(\APPROX(s)) = \rev_D(s)$$
Still by \Cref{lemma:cost-1}, the externalities are at most
\begin{align*}
    \frac{\ext_D(\APPROX(s))}{\ext_D(s)} &\leq \frac{e^{-c}}{E_{t \leftarrow D}[\risk(t, s) | \pvalue(t, s) \geq p(s)]}\\
    &\leq \frac{e^{-c}}{E_{t \leftarrow D}[\risk(t, s) \cdot \indicator{k < k_0(s)}| \pvalue(t, s) \geq p(s)]}\\
    &= \frac{e^{-c}}{\varepsilon_1 e^{-c}} = \frac{1}{\varepsilon_1} \leq 8
\end{align*}

If $\APPROX(s)$ outputs $\COST^1(s, \varepsilon_1 + \varepsilon_2)$, then by \Cref{corollary:eps-2},
\begin{align*}
    \rev_D(\APPROX(s)) \geq (\varepsilon_1 + \varepsilon_2) \rev_D(s) \geq \frac{1}{8}\rev_D(s)
\end{align*}
\begin{align*}
    \ext_D(\APPROX(s)) \leq \frac{1}{\varepsilon_2 e} \ext_D(s) \leq \frac 8 e \ext_D(s)
\end{align*}
Next, \Cref{claim:fine-approx} proves approximation guarantees when $\APPROX(s)$ outpus $\FINE(s)$.

\begin{claim}\label{claim:fine-approx}
\begin{align*}
    \rev_D(\FINE(s)) \geq \frac{3}{16} \rev_D(s)
\end{align*}
\begin{align*}
    \ext_D(\FINE(s)) \leq 24 \ext_D(s)
\end{align*}
\end{claim}
\begin{proof}
If $\APPROX(s)$ outputs $\FINE(s)$ then $\varepsilon_3 \geq 3/4$. Let's first bound the profit of $\FINE(s)$. If $\FINE(s)$ outputs $\INV(s, p, p/(p-c))$, then it must be $c = 1$. By \Cref{lemma:inv-transform},
\begin{align*}
    \rev_D(\FINE(s)) = \frac{p}{p-c}\rev_D(s) \geq \rev_D(s)
\end{align*}
If $\FINE(s)$ outputs $\BLOWUP \circ \INV(s, p, 1/2)$, by \Cref{claim:blowup-profit}, and \Cref{lemma:inv-transform},
\begin{align*}
    \rev_D(\FINE(s)) \geq \frac{1}{2}\beta \varepsilon_3 \rev_D(s) \geq \frac{3}{16} \rev_D(s)
\end{align*}
If $\FINE(s)$ outputs $\HEAVY(s)$, $\COST^1(s, 1)$ or $\COST^x(s)$, then by \Cref{claim:fine-cost-1}, \ref{claim:cost-3-good}, and \Cref{lemma:cost-1},
\begin{align*}
    \rev_D(\FINE(s)) \geq \beta \varepsilon_3 \rev_D(s) \geq \frac{3}{8} \rev_D(s)
\end{align*}
Let's now bound the externality of $\FINE(s)$. If $\FINE(s)$ outputs $\COST^1(s, 1)$, then $E_{t \leftarrow D}[1/yk \cdot \indicator{k > k_h(s)} | A(s)] \geq \frac{e^{-c}\varepsilon_3(1-\beta)}{4} $ and by \Cref{lemma:cost-1} 
\begin{align*}
    \frac{\ext_D(\FINE(s))}{\ext_D(s)} &\leq \frac{e^{-c}}{E_{t \leftarrow D}[\risk(t, s) | \pvalue(t, s) \geq p(s)]}\\
    &\leq \frac{e^{-c}}{E_{t \leftarrow D}[1/yk \cdot \indicator{k > k_h(s)} | A(s)]}\\
    &\leq \frac{4}{\varepsilon_3(1-\beta)} \leq \frac{32}{3}
\end{align*}
If $\FINE(s)$ outputs $\INV(s, p, p/(p-c))$, then it must be $c \leq 1$ and by \Cref{lemma:inv-transform},
\begin{align*}
    \frac{\ext_D(\FINE(s))}{\ext_D(s)} = e^{c} \leq e
\end{align*}
If $\FINE(s)$ outputs $\HEAVY(s)$, by \Cref{claim:blowup-bad},
\begin{align*}
    \frac{\ext_D(\FINE(s))}{\ext_D(s)} \leq \frac{4}{(1-\beta)(c + 1)\varepsilon_3} \leq \frac 8 3
\end{align*}
If $\FINE(s)$ outputs $\COST^3_x(s)$, by \Cref{claim:cost-3-good},
\begin{align*}
    \frac{\ext_D(\FINE(s))}{\ext_D(s)} \leq \frac{2}{(1-\beta)\varepsilon_3} = \frac{16}{3}
\end{align*}
If $\FINE(s)$ outputs $\BLOWUP \circ \INV(s, p, 1/2)$, by \Cref{claim:blowup-good}, \ref{claim:cost-3-bad-blowup} and \Cref{lemma:inv-transform},
\begin{align*}
    \frac{\ext_D(\FINE(s))}{\ext_D(s)} &\leq e^{-1/2(p-c)}\max\bigg\{\frac{4 (p - \frac{1}{2}c - \frac{1}{2}p)}{(1-\beta)(1+\frac{1}{2}c + \frac{1}{2} p)\varepsilon_3} + \frac{1}{\beta\varepsilon_3}\\
    &, \frac{4 (p - \frac{1}{2}c - \frac{1}{2}p) }{(1-\beta)\varepsilon_3} + \frac{1}{\beta\varepsilon_3} \bigg\}\\
    &\leq \frac{4 (p - c)}{(p-c)(1-\beta)\varepsilon_3} + \frac{1}{\beta\varepsilon_3}\\
    &=\frac{4}{(1-\beta)\varepsilon_3} + \frac{1}{\beta \varepsilon_3} \leq 40/3
\end{align*}
In the worst case, the profit is $3/16$ of $\rev_D(s)$ and the externality is $40/3$ times higher than $\ext_D(s)$ which concludes the proofs.
\end{proof}

To conclude, for the profit, in the worst case, $\APPROX(s)$ outputs $\COST^2(s)$ at a compromise of at most $1/8$ of the profit. For the externalities, $\FINE(s)$ will have $40/3$ times more externalities than $s$ which completes the proof.
\end{proof}
\section{Profits-Maximizing Seller}\label{sec:profits-maximizing-seller}
In this section, we consider the a variant of our main model. Here, the seller always selects the profits-maximizing price given $y$ and $c$, and we also compute externalities differently. Specifically, we consider the three stage game where the regulator commits on a fine $y$ and cost $c$ and in sequence the seller is free to select the profit optimal price $p$. In the last step, the buyer decides to purchase or not. Finally, we will assume externalities are measured as the expected risk conditioned on a buyer to purchase times the probability of purchase (i.e. the total probability of compromise, versus compromise conditioned on purchase). One purpose of this section is to explore a variant of our model. The other purpose is to highlight that results do not significantly change in related models.

\begin{definition}[Externality]
Given a policy $s=(y, c, p)$, we define
\begin{align*}
\ext_D(s) := E_{t \leftarrow D}[\risk(t, s) \cdot \indicator{\pvalue(t, s) \geq p}]
\end{align*}
\end{definition}

Observe that in \Cref{ex:lower-bound}, the impossibility result is shown by deriving a distribution $D$ and profit constraint $R$ that can only be satisfied if everyone purchase; therefore, the impossibility follows to the profits-maximizing case.

Next, we will proof \Cref{prop:newexample} for the profits-maximizing seller.

\begin{proposition}\label{prop:homogeneous-counterexample-2}
For a profits-maximizing seller, there exist distributions $D_v, D_k$, and profits constraint $R$ such that:
\begin{itemize}
\item $T$ is such that for all $k \geq T$, the externality-minimizing policy for $D_v\times\{k\}$ subject to profits constraints $R$ is a fine policy. 
\item $D_k$ is supported on $[T,\infty)$.
\item \emph{No} fine policy is externality-minimizing policy for $D_v \times D_k$ subject to profits constraints $R$.
\end{itemize}
\end{proposition}

Under profit constraint $R$, we can show that the optimal policy gives profits $R$, \Cref{corollary:profit-constraint-tight}. Before we proof this result, we will extend the invariant property, \Cref{lemma:inv-transform}, to the profits-maximizing setting.

\begin{lemma}[Augmented Invariant Property]\label{lemma:invariant-improvement}
Given policy $s$, let $s' = \INV(s, \alpha)$, $\alpha \in [0, 1]$, then $\rev_D(s') \geq \alpha \rev_D(s)$, and the optimal price of $s'$ is $p' \geq p$ which implies $\ext_D(s') \leq e^{-(p-c)(1-\alpha)} \ext_D(s)$.
\end{lemma}
\begin{proof}
Let $p(s)$ denote the optimal price under policy $s$. Let $c(s') = \alpha c + (1-\alpha) p$. It follows,
\begin{align*}
(p(s) - c) Pr_{t \leftarrow D}[\pvalue(t, s) \geq p(s)] &\geq (p(s') - c) Pr_{t \leftarrow D}[\pvalue(t, s) \geq p(s')]\\
(p(s') - c(s')) Pr_{t \leftarrow D}[\pvalue(t, s') \geq p(s')] &\geq (p(s) - c(s')) Pr_{t \leftarrow D}[\pvalue(t, s') \geq p(s)]
\end{align*}
By adding the inequalities and observing that by \Cref{lemma:inv-transform}, for all price $p$, $Pr_{t \leftarrow D}[\pvalue(t, s) \geq p] = Pr_{t \leftarrow D}[\pvalue(t, s') \geq p]$, we have
\begin{align*}
&(p(s) - p(s) + c(s') - c) Pr_{t \leftarrow D}[\pvalue(t, s) \geq p(s)]
\geq (p(s') - p(s') + c(s') - c) Pr_{t \leftarrow D}[\pvalue(t, s) \geq p(s')]\\
&\iff (c(s') - c) Pr_{t \leftarrow D}[\pvalue(t, s) \geq p(s)] \geq (c(s') - c) Pr_{t \leftarrow D}[\pvalue(t, s) \geq p(s')]\\
&\iff Pr_{t \leftarrow D}[\pvalue(t, s) \geq p(s)]  \geq Pr_{t \leftarrow D}[\pvalue(t, s) \geq p(s')]
\end{align*}
In the last step, if the profit is non-zero then $p > c$ which implies $c(s') > c$ because $c(s')$ is a convex combination of $p$ and $c$. We conclude $p(s) \leq p(s')$. Since $p(s')$ gives more profit than $p(s)$, we must have $\rev_D(s') \geq \alpha \rev_D(s)$. In addition, the fact the distribution of values is identical and we only sell with lower probability implies $\ext_D(s') \leq e^{-(p-c)(1-\alpha)}\ext_D(s)$ which completes the proof.
\end{proof}

\begin{corollary}\label{corollary:profit-constraint-tight}
If policy $s$ is optimal, then $\rev_D(s) = R$.
\end{corollary}
\begin{proof}
Suppose $\rev_D(s) > R$, then there is a sufficiently small $\varepsilon > 0$ such that $s' =\INV(s, (1-\varepsilon))$ and $\rev_D(s') \geq R$ by \Cref{lemma:invariant-improvement}. Yet by \Cref{lemma:invariant-improvement} $\ext_D(s') < \ext_D(s)$ contradicting $s$ was optimal.
\end{proof}

\paragraph{Notation}
For arbitrary distribution $D$, let
$$\supp(D) := \{x_1, ..., x_n\}$$
then for $i \in [|\supp(D_v)|]$, $j \in [|\supp(D_k|]$, let $t_{ij} := (v_i, k_j)$,
$$r_{ij}(s) := \pvalue(t_{ij}, s) Pr_{t \leftarrow D}[\pvalue(t, s) \geq \pvalue(t_{ij}, s)]$$
In the discrete case, we have $\rev_D(s) := \max_{i, j} r_{ij}(s)$. To compute $r_{ij}(s)$, we must compute if for some $i' \in [|\supp(D_v)|]$, $j' \in [|\supp(D_k|]$, $\pvalue(t_{ij}, s) < \pvalue(t_{i'j'}, s)$ or $\pvalue(t_{ij}, s) \geq \pvalue(t_{i'j'}, s)$.

\begin{definition}[Type Total Order]
We define a type total order $<_s$ where $t_{ij} <_s t_{i'j'}$ denote $\pvalue(t_{ij}, s) \leq \pvalue(t_{i'j'}, s)$. Observe that for all regulation $s$, fix indexes $i$ and $j$, then for all $i' \geq i$, $j' \geq j$, $t_{ij} <_s t_{i'j}$, $t_{ij} <_s t_{ij'}$.
\end{definition}

Before we proof \Cref{prop:homogeneous-counterexample-2}, let's get intuition why one might expect it to be true. Assume $D_v$ and $D_k$ have each support of cardinality 2. When there is no regulation $s_0 = (0, 0)$, we have the total order of types
\begin{align*}
t_{11} <_{s_0} t_{12}<_{s_0} t_{21} <_{s_0} t_{22}
\end{align*}
As we increase fines, for some $y > 0$ and policy $s = (y, 0)$, we will have $\pvalue(t_{12}, s) = \pvalue(t_{21}, s)$ since $t_{12}$ is more efficient than $t_{21}$. However, assume at the highest possible fine, the seller prefers to sell to $\{t_{12}, t_{21}, t_{22}\}$. When fines are slightly smaller, the seller might prefer to sell only to $t_{21}$, and $t_{22}$ if $v_2$ is sufficiently larger than $v_1$. If that is the case, selling at lower probability might be better to decrease externalities than increasing fines. If that is the case, the optimal fine policy will give more profit than the profit constraint $R$, and by \Cref{corollary:profit-constraint-tight} such policy cannot be optimal.

\begin{proof}[Proof of Proposition \ref{prop:homogeneous-counterexample-2}]
Define the value distribution $D_v$, and efficiency distribution $D_k$,
\begin{align*}
D_v = \begin{cases}
v_1 = 1 & \text{w. p. $\frac{1}{2}$}\\
v_2 = 1.58 & \text{w. p. $\frac{1}{2}$}
\end{cases}
\end{align*}
\begin{align*}
D_k = \begin{cases}
k_1 = 3 & \text{w. p. $\frac{1}{2}$}\\
k_2 = 9 & \text{w. p. $\frac{1}{2}$}
\end{cases}
\end{align*}
Consider fines of the form $y(k) := \frac{e^{\frac{k}{k-1}}}{ek}$.

Let the profit constraint be $R = r_{12}(1.2 \cdot y(3), 0)$. We will have $R \in (0.51, 0.52)$, but to compute $R$, we must first compute the typs total order when fines are $1.2 \cdot y(3)$, \Cref{claim:example-total-order}. We will then show that setting externalities to $y(3)$ yields strictly lower externalities then setting fines to $1.2 \cdot y(3)$, \Cref{claim:profits-maximizing-example}. Further we argue that when fines are $1.2 \cdot y(x)$, $D_v \times \{3\}$ has optimal fine policy, \Cref{claim:threshold-assumption}.

\begin{claim}\label{claim:example-total-order}
If $s = (y, 0)$, $y \leq 1.2 \cdot y(3)$ then $t_{11} <_s t_{12}<_s t_{21} <_s t_{22}$.
\end{claim}
\begin{proof}
Clearly $t_{11} <_s t_{12}$ and $t_{21} <_s t_{22}$. It is sufficient to show for $y = 1.2 \cdot y(3)$, $t_{12} <_s t_{21}$ and the claim follows for all fines smaller than $1.2 \cdot y(3)$ because $\frac{\partial^2 \pvalue(t, s)}{\partial y \partial k} \geq 0$. Observe,
\begin{align*}
\pvalue(t_{21}, s) - \pvalue(t_{12}, s) &= 1.58 - 1 + \ell(3, s) - \ell(9, s)\\
&> 0.83 > 0
\end{align*}
which completes the proof.
\end{proof}

Let's now proof when the fines are in the range $[y(3), 1.2 \cdot y(3)]$, the seller never sells with probability 1.

\begin{claim}
For $y \in [y(3), 1.2 \cdot y(3)]$, then
\begin{align*}
r_{21}(y, 0) > r_{11}(y, 0)
\end{align*}
\end{claim}
\begin{proof}
By \Cref{claim:example-total-order}, if $y \leq 1.2 \cdot y(3)$, then $t_{11} <_s t_{12}<_s t_{21} <_s t_{22}$. We can then compute the profit,
\begin{align*}
r_{21}(y, 0) - r_{11}(y, 0) &= \frac{1}{2}(1.58 - \ell(3, y)) - 1 + \ell(3, y)\\
&> -0.25 + \frac{1}{2}\ell(3, y)\\
&\geq -0.25 + \frac{1}{2}\ell(3, y(3))\\
&= 0
\end{align*}
which completes the proof.
\end{proof}

Next, we claim if the fine is $1.2\cdot y(3)$, the seller prefers to sell to $t_{12}$, $t_{21}$, $t_{22}$ with probability $\frac 3 4$ and when the fine is $y(3)$, the seller sells only to $t_{21}$ and $t_{22}$.

\begin{claim}\label{claim:seller-threshold}
$r_{21}(y(3), 0) > r_{12}(y(3), 0)$ and $r_{12}(1.2\cdot y(3), 0) > r_{21}(1.2\cdot y(3), 0)$
\end{claim}
\begin{proof}
For arbitrary $y$,
\begin{align*}
r_{21}(y, 0) - r_{12}(y, 0) = \frac{1}{2}(1.58 - \ell(3, y)) - \frac{3}{4}(1-\ell(9, y))
\end{align*}
By computing the left hand side for $y = y(3)$ and $y = 1.2 \cdot y(3)$, we have $r_{21}(y(3), 0) - r_{12}(y(3), 0) > 0$ and $r_{21}(1.2\cdot y(3), 0) - r_{12}(1.2 \cdot y(3), 0) < 0$.

\end{proof}

Next, we show that having slightly lower fines yields lower externalities.

\begin{claim}\label{claim:profits-maximizing-example}
\begin{align*}
\ext_D(y(3), 0) < \ext_D(1.2\cdot y(3), 0)
\end{align*}
\end{claim}
\begin{proof}
For all fines $y$, if the seller sells at price $\pvalue(t_{22}, y, 0)$, the profit is at most $1/4$ which is smaller than the profit constraint $R > 1/2$ . We have shown that for $y \in [y(3), 1.3\cdot y(3)]$ the seller will not prefers to sell to everyone; therefore, the seller either sells to $\{t_{12}, t_{21}, t_{22}\}$ or to $\{t_{21}, t_{22}\}$. In particular, when $y = y(3)$, the seller sells to $t_{21}$ and $t_{22}$ only. The externalities are
\begin{align*}
Ext_D(y(3), 0) &= \frac{1}{4 y(3)}\bigg(\frac{1}{3} + \frac{1}{9}\bigg)\\
&< 0.203
\end{align*} 
When $y = 1.2 \cdot y(3)$, the seller sells to $t_{12}$, $t_{21}$, $t_{22}$ and the externalities are
\begin{align*}
Ext_D(1.2 \cdot y(3), 0) &= \frac{1}{4 \cdot 1.2 \cdot y(3)}\bigg(\frac{1}{3} + \frac{2}{9}\bigg)\\
&> 0.21
\end{align*}
\end{proof}

Observe $1.2 \cdot y(3)$ is the highest fine of any feasible fine policy by our profit constraint $R = r_{12}(1.2 \cdot y(3), 0)$. This is because $r_{12}(1.2 \cdot y(3), 0)$ is the highest profit we can get under policy $(1.2 \cdot y(3), 0)$. For fines $y$ strictly larger than $1.2 \cdot y(3)$, $\pvalue(t, 1.2 \cdot y(3), 0)$ strictly stochastic dominate $\pvalue(t, y, 0)$ implying the profit is strictly smaller than $R$. By the same argument, policy $(y(3), 0)$ gets profit $> R$, and by \Cref{corollary:profit-constraint-tight}, there is a non-simple policy that yields lower externalities than $(y(3), 0)$ which implies a non-fine policy is optimal. We conclude by claiming the distribution $D_v \times \{3\}$ with profit constraint $R$ has optimal fine policy.

\begin{claim}\label{claim:threshold-assumption}
A fine policy is optimal on $D_v \times \{3\}$, $R$.
\end{claim}
\begin{proof}
Another way to state \Cref{lem:hardestcase} is to observe that given a profit constraint $R$, if $\frac{1}{T - 1}$ is the highest loss we can have, then $T$ is the threshold where for all $D_v \times \{k\}$, $k \geq T$, the optimal policy is a fine policy. Observe that $\ell(k, y(k), 0) = \frac{1}{k-1}$. This implies, if $y(3)$ is feasible on $(D_v \times \{3\}, R)$ then a fine policy is optimal.

Let $s = (y(3), 0)$ and suppose the seller sells at price $1.58 - \ell(3, y(3), 0)$ having a sale with probability $\frac{1}{2}$ and getting profit $0.54$. We have $R < 0.52$; therefore, the policy $s$ is feasible. By definition of $y(k)$, since $(y(3), 0)$ is feasible, the threshold $T$ where fine policies are optimal is at least $3$. We can conclude, with distribution $D_v \times \{3\}$ and profit constraint $R$, we have an optimal fine policy with fine at least $y(3)$.
\end{proof}
this completes the proof of \Cref{prop:homogeneous-counterexample-2}.
\end{proof}}
\end{document}